%% file: main.tex
\documentclass[article]{llncs}
\usepackage{etoolbox} 

\newtoggle{llncs}
\makeatletter
\@ifclassloaded{llncs}{\toggletrue{llncs}}{\togglefalse{llncs}}
\makeatother

\iftoggle{llncs}{
\input{settingsllncs}}{
\input{settings}
}

\input{macros.tex}

\title{MIRANDA: short signatures from a leakage-free full-domain-hash scheme}
\author{Alain Couvreur\inst{1}\and Thomas Debris-Alazard\inst{1} \and Philippe Gaborit\inst{2}\and Adrien Vinçotte\inst{3}}
\institute{Inria and  Laboratoire LIX, \'Ecole Polytechnique, Palaiseau, France \and XLIM, University of Limoges, Limoges, France \and IRMAR, University of Rennes, Rennes, France}

\begin{document}		
	
\maketitle

    \input{src/abstract}
	\input{src/intro}
	\input{src/preliminaries}
	\input{src/signature_via_atpsf}

	\input{src/signature_scheme}
	\input{src/security}
	\iftoggle{llncs}{}{
        \input{src/security_moved}
      }
	\input{src/parameters}
	\input{src/encryption_scheme}
    \input{src/acknowledgements}
	\iftoggle{llncs}{}{
	\input{src/conclusion}

}
	\iftoggle{llncs}{
	  \bibliographystyle{src/splncs04}
	}{
    \bibliographystyle{alpha}
      }      
      \bibliography{codecrypto,biblio}

	\newpage
	\appendix
        \input{src/appInvertF}
        \input{src/app_proof_average_time}
        \input{src/appProofLHL}
        \input{src/app_prob_rank}

\input{src/security_moved}
\end{document}

%% file: settingsllncs.tex
\usepackage{amsmath}
\usepackage{amssymb}
\usepackage{bbm}
\usepackage[english]{babel}
\usepackage{color}
\usepackage{enumitem}
\usepackage{ellipsis}
\usepackage[breaklinks,colorlinks=true, citecolor=blue]{hyperref}
\usepackage[misc,geometry]{ifsym} 
\usepackage{dsfont}
\usepackage{mathrsfs}
\usepackage{mathtools}
\usepackage{algorithm}
\usepackage{algpseudocode}
\usepackage{stmaryrd}
\usepackage{thm-restate}
\usepackage{subcaption}
\usepackage{systeme}

\usepackage{tikz}
\usepackage{accents}
\usetikzlibrary{patterns}
\usetikzlibrary{matrix}

\usepackage{todonotes}
\usepackage{lipsum}

\pagestyle{plain}

%% file: settings.tex
\pagestyle{plain}
\usepackage[utf8]{inputenc}
\usepackage{a4wide}
\usepackage{amsmath}
\usepackage{amssymb}
\usepackage[english]{babel}
\usepackage{color}
\usepackage{enumitem}
\usepackage[breaklinks,colorlinks=true, citecolor=blue]{hyperref}
\usepackage{dsfont}
\usepackage{mathrsfs}
\usepackage{mathtools}
\usepackage{algorithm}
\usepackage[noend]{algpseudocode}
\usepackage{graphicx}
\usepackage{stmaryrd}
\usepackage{thm-restate}
\usepackage[normalem]{ulem}
\usepackage{tikz}
\usetikzlibrary{matrix}
\usetikzlibrary{arrows}
\usetikzlibrary{patterns, shapes, positioning}
\usetikzlibrary{decorations.pathmorphing,shapes}
\usepackage{makecell, multirow} 
\usepackage{setspace}
\usepackage{xspace}

\usepackage{comment}
\usepackage{bbm}

\usepackage[foot]{amsaddr}


%% file: macros.tex
\newcommand{\F}{\mathbb{F}_2}

\newcommand{\Fq}{\mathbb{F}_q}

\newcommand{\Fqm}{\mathbb{F}_{q^m}}

\renewcommand{\vec}[1]{\mathbf{#1}}
\newcommand{\av}{\vec{a}}

\newcommand{\gv}{\vec{g}}
\newcommand{\sv}{\vec{s}}
\newcommand{\yv}{\vec{y}}
\newcommand{\vv}{\vec{v}}
\newcommand{\Av}{\vec{A}}
\newcommand{\Bv}{\vec{B}}
\newcommand{\Cv}{\vec{C}}
\newcommand{\Dv}{\vec{D}}
\newcommand{\Ev}{\vec{E}}

\newcommand{\Hv}{\vec{H}}
\newcommand{\Mv}{\vec{M}}
\newcommand{\Nv}{\vec{N}}
\newcommand{\Pv}{\vec{P}}
\newcommand{\Qv}{\vec{Q}}

\newcommand{\Vv}{\vec{V}}
\newcommand{\Wv}{\vec{W}}
\newcommand{\Xv}{\vec{X}}

\newcommand{\minrank}{\mathsf{MinRank}}

\iftoggle{llncs}{
	\spnewtheorem{fact}{Fact}{\bfseries}{\itshape}}{
	\newtheorem{remark}{Remark}
	\newtheorem{lemma}{Lemma}
	
	\newtheorem{definition}{Definition}
	\newtheorem{theorem}{Theorem}
	\newtheorem{proposition}{Proposition}

}

\newcommand{\eqdef}{\stackrel{\textup{def}}{=}}

\DeclareMathOperator{\rk}{rk}
\DeclareMathOperator{\tr}{Tr}

\DeclareMathOperator{\prob}{\mathbb{P}}
\DeclareMathOperator{\Prob}{\mathbb{P}}

\newcommand{\Gabcode}[2]{\ensuremath{\mathsf{Gab}_{#1}(#2)}}

\newcommand{\miranda}{\mathsf{Miranda}}

\newcommand{\E}{\ensuremath{\mathcal{E}}\xspace}

\newcommand{\A}{\ensuremath{\mathcal{A}}\xspace}

\newcommand{\negl}{\normalfont\textsf{negl}}

\newcommand{\Amat}{\mathcal A_{mat}}
\newcommand{\Cmat}{\mathcal C_{mat}}
\newcommand{\Dmat}{\mathcal D_{mat}}

\newcommand{\Rs}{\mathcal R_{s}}

\newcommand{\C}{\mathcal C}

\newcommand{\map}[4]{\left\{
    \begin{array}{rcl}
      #1 & \longrightarrow & #2 \\[1pt] #3 & \longmapsto & #4
    \end{array}
  \right.
}
\renewcommand{\geq}{\geqslant}
\renewcommand{\leq}{\leqslant}

\newcommand{\addremove}{\ensuremath{\mathsf{AddRemove}}\xspace}

\DeclareMathOperator*{\trap}{\text{\tt{Trapdoor}}}
\DeclareMathOperator*{\sampPre}{\text{\tt{InvertAlg}}}
\newcommand{\sk}{\mathrm{sk}}
\newcommand{\pk}{\mathrm{pk}}
\DeclareMathOperator*{\hash}{\text{\tt{Hash}}}
\newcommand{\Sgnsk}{\ensuremath{\mathtt{Sgn}}}
\newcommand{\Vrfypk}{\ensuremath{\mathtt{Vrfy}}}
\newcommand{\fow}{\ensuremath{f_{\mathsf{OW}}}}

\newcommand{\Unif}{\hookleftarrow}


%% file: src/abstract.tex
	\begin{abstract}
      We present $\miranda$, the first family of full-domain-hash
      signatures based on matrix codes. This signature scheme fulfils
      the paradigm of {\it Gentry, Peikert and Vaikuntanathan}
      ($\mathsf{GPV}$), which gives strong security guarantees. Our
      trapdoor is very simple and generic: if we propose it with
      matrix codes, it can actually be instantiated in many other ways
      since it only involves a 
decodable code (or lattice) in a unique
      decoding regime of parameters.  Though $\miranda$ signing
      algorithm relies on a decoding task where there is exactly one
      solution, there are exponentially many possible
      signatures given a message to sign and we ensure that signatures
      are not leaking information on their underlying trapdoor. The
      latter fact rests on a very simple procedure involving the
      drawing of a small amount of uniformly random bits. In particular
      $\miranda$ does not use a rejection sampling procedure
      or floating arithmetic 
      which
      makes its implementation a very simple task contrary to other
      $\mathsf{GPV}$-like signatures schemes such as $\mathsf{Falcon}$
      or even $\mathsf{Wave}$.
			
      We instantiate $\miranda$ with the famous family of Gabidulin codes represented as spaces of matrices, and we study thoroughly its security (in the EUF-CMA security model). We propose some NIST-compliant parameters for $\miranda$: the signature sizes are as low as~$59$ bytes and the public key sizes are in the order of~$734$ kilobytes for the level I.

        	Finally, as a by-product we also show by tuning the parameters how to turn our paradigm into an encryption scheme. 
	\end{abstract}


%% file: src/intro.tex
\section{Introduction}

{\bf \noindent Signatures and quantum resistant cryptography.} Until recently, few resistant signatures schemes against quantum attacks were known. NIST has made no mistake about this: after launching an initial call in 2017 for the standardisation of quantum resistant primitives, which has now ended, they launched a second call in 2023 to standardise additional signature schemes with the aim of diversifying security assumptions. Among the potential quantum resistant assumptions, multivariate based Full Domain Hash (FDH) signature schemes through the $\mathsf{UOV}$ approach~\cite{uov} are particularly suitable for obtaining very small signature sizes but lead to rather large public keys. Based on this approach, $\mathsf{Mayo}$~\cite{mayo} proposes a very efficient trade-off between signature and public key sizes. However, its security relies on new assumptions which are still under the scrutiny of the community.

In this paper, we are interested by instantiating a signature scheme comparable to $\mathsf{Mayo}$ in terms of efficiency and trade-off between signature and public key sizes but whose security relies on the hardness of decoding a random linear matrix code. 
\medskip

{\bf \noindent Full-Domain Hash (FDH) signatures and code-based cryptography.} It has been a long-standing open problem to build a secure and efficient FDH signature scheme whose security inherits from the hardness of decoding a random linear code. The first answer to this question was provided by $\mathsf{CFS}$~\cite{CFS01}. It consisted in finding parity-check matrices $\vec{H} \in \mathbb{F}_{2}^{(n-k) \times n}$ such that the solution~$\vec{e}$ of smallest weight of the equation
\begin{equation}\label{eq:CFS}
	\vec{e}\vec{H}^{\top} = \vec{s} 
\end{equation}
could be found for a non negligible proportion of all $\vec{s} \in \mathbb{F}_{2}^{n-k}$.
The latter $\vec{s}$ being the hash of the message to sign.
This task was achieved by using high rate Goppa codes. These codes were proposed for two main reasons. First, Goppa codes admit an efficient decoding algorithm in their unique decoding regime of parameters. Second, choosing high rate Goppa codes, {\it i.e.,}~$k \approx n$, ensures that a significant proportion of $\vec{s}$ equals some~$\vec{e}\vec{H}^{\top}$ for some short vector $\vec{e}$. 
	Unfortunately this second strong requirement has recently been exploited to break $\mathsf{CFS}$. It has been shown in \cite{FGOPT13} that high rate Goppa codes are not suitable for cryptographic applications. However, it is important to note that despite this polynomial-time attack, the use of a decoding algorithm in its unique decoding regime of parameters has enabled $\mathsf{CFS}$ to offer particularly short signatures by design (originally in the order of $50$ bytes for $128$ bits of security).

In order to circumvent the shortcomings of $\mathsf{CFS}$, another opposite approach was proposed
via 
$\mathsf{Wave}$~\cite{DST19a}. In this scheme, we
seek to solve \eqref{eq:CFS} with $\vec{e} \in \mathbb{F}_3^n$ having a large
enough Hamming weight, but not by relying on an underlying decoding
algorithm in its unique decoding regime of parameters. It was instead
proposed to relax 
what we ask for the code and its associated
``decoding'' algorithm. Namely, we ask for a code such that we can
solve~\eqref{eq:CFS} in a regime of parameters where it has many
solutions.
This kind of code is not used for  error correction but can be found in lossy source coding or source distortion theory where the problem is to find codes with an associated decoding algorithm which can approximate {\em any} word of the ambient space by a close enough codeword. It turns out that when this approach was proposed in code-based cryptography \cite{DST17,DST17a}, 
the latter was already known in the lattice-based context~\cite{GPV08} via the $\mathsf{GPV}$ framework. In particular, it was known that choosing parameters where~\eqref{eq:CFS} admits plenty of solutions could potentially be disastrous for the security of the underlying signature scheme. Indeed, when solving~\eqref{eq:CFS}, we have to be \emph{extremely} careful: solutions of this equation are intended to be made public (these are signatures), while to compute them we precisely use a trapdoor, namely some underlying secret structure (like a short basis in the lattice case). It is therefore crucial to ensure that the distribution of solutions is independent of the used trapdoor. To ensure the security, the $\mathsf{GPV}$ framework has then consisted of the requirement to design an algorithm computing solutions to \eqref{eq:CFS} and whose output distribution is independent from the used trapdoor. This is the so-called design of \emph{trapdoor preimage sampleable functions}. This elegant framework allows the design of secure FDH signature schemes, but this comes at a certain cost: it is generally very difficult to design an algorithm that solves~\eqref{eq:CFS} with the prerequisite mentioned above. Usually these algorithms rely on the use of a so-called \emph{rejection sampling phase} which implies many implementation difficulties and opens a possible wide range of side channel attacks like in the case of $\mathsf{Falcon}$~\cite{FHKLPPRSWZ} and $\mathsf{Wave}$~\cite{DST19a}. Furthermore, as solutions of~\eqref{eq:CFS} are not as short as in $\mathsf{CFS}$, it leads to larger signature sizes (in the order of $700$ bytes in $\mathsf{Falcon}$ and $\mathsf{Wave}$). 
\medskip

{\bf \noindent Our first contribution: $\miranda$ a new framework to design preimage sampleable functions.} Motivated by this current state of the affairs, we have wished to take advantage of the best of both worlds, namely $\mathsf{CFS}$ and $\mathsf{GPV}$, without their inherent drawbacks. This is why we aimed to solve  \eqref{eq:CFS} in a regime of parameters where there are many solutions. Our idea is rather simple, we aim to decode a code $\mathcal{D}$ (which is publicly known) by using some underlying structure. To this aim we first select a 
some code~$\mathcal{C}$ that we know how to efficiently decode in its unique decoding regime of parameters. Then we add to $\mathcal{C}$ 
a random code~$\mathcal{A}$ to form $\mathcal{D}$, {\it i.e.,}
$$
\mathcal{D} \eqdef \mathcal{C} \oplus \mathcal{A}~.
$$
Here we did 
a crucial operation to avoid key recovery attacks: we {\it added} a code~$\mathcal{A}$ to hide the (secret) structure of $\mathcal{C}$. 
Then to decode $\mathcal{D}$ we proceed as follows: given a target $\vec{y}$ to decode, we first pick uniformly at random $\vec{a} \in \mathcal{A}$ and then we try to decode $\vec{y}-\vec{a}$ in~$\mathcal{C}$
by using 
its decoding algorithm
{\em in its unique decoding regime of parameters}. Our rationale is that when $\mathcal{A}$ is large enough then there will almost always exist $\vec{a}$ such that 
\begin{equation}\label{eq:decoding_formula}
\vec{y} - \underbrace{\vec{a}}_{\in \mathcal{A}} = \underbrace{\vec{c}}_{\in \mathcal{C}} + \vec{e}
\end{equation}
for some short enough $\vec{e}$. The key of our signature is that the
target weight for~$\vec{e}$ is in the unique decoding regime for
$\mathcal{C}$
but turns out to lie far above the Gilbert--Varshamov radius for
$\mathcal D$. Thus, if for a given $\av$ there is at most one solution
$\vec{e}$ for (\ref{eq:decoding_formula}), there are exponentially many
possible choices for $\av$ such that (\ref{eq:decoding_formula}) has a
solution.

Notice that here we picked $\vec{a} \in \mathcal{A}$ uniformly at random. This will ensure that the solution $\vec{e}$ we output will be uniformly random among the whole set of possible solutions. Our fundamental remark is that if we pick a correct $\vec{a}$, then it cannot exist distinct 
$\vec{c},\vec{c}'\in \mathcal{C}$ 
and $\vec{e}, \vec{e}'$ short enough such that
~$\vec{y}-\vec{a} = \vec{c} + \vec{e} = \vec{c}' + \vec{e}'$. 
Otherwise, this would contradict the fact that we decode~$\mathcal{C}$
in its unique decoding regime of parameters. It shows that given some fixed~$\vec{y}$, we compute some possible solution $\vec{e}$ uniformly at random. However, this is not yet sufficient. It might be possible that almost all the short vectors $\vec{e}$ are solutions for few targets~$\vec{y}$. It would imply that our decoding algorithm concentrates its outputs $\vec{e}$ on a small subset of short vectors (as $\vec{y}$ being the hash of the message to be signed is considered as uniform). Fortunately, using the leftover hash lemma (by randomizing over the code $\mathcal{D}$), we can prove that this does not hold. As a conclusion, all targets~$\vec{y}$ have essentially the same amount of solutions $\vec{e}$. Therefore, our algorithm computes a solution~$\vec{e}$ following a \emph{uniformly random} distribution over short vectors. This shows that signatures are not leaking information on their trapdoor (here the codes $\mathcal{A}$ and 
$\mathcal{C}$). Notice that our procedure, called $\miranda$, to achieve this is extremely simple, it simply amounts to draw $\vec{a}\in \mathcal{A}$ uniformly at random until $\yv - \av$ is a valid input of a decoding algorithm for 
$\mathcal{C}$
in its unique decoding regime of parameters. In particular, we did not use a tedious rejection sampling phase: we output a solution as soon as we pick a correct $\vec{a} \in \mathcal{A}$; we never reject a correct solution, {\it i.e.,} a short enough $\vec{e}$ such that~$\vec{y} - \vec{e} \in \mathcal{C}\oplus \mathcal{A}$.

The fact that signatures are not leaking information is the key point to prove that our $\miranda$-construction is secure 
when assuming the hardness to find a close codeword in $\mathcal{D}=\mathcal{C} \oplus \mathcal{A}$ from a random target $\vec{y}$ and a public description of~$\mathcal{D}$. In particular, this description is not revealing codes $\mathcal{C}$ and $\mathcal{A}$. It is important to note here that we do not reduce the security of $\miranda$ to the hardness of decoding a random code. What can be said instead is that $\miranda$ is secure if decoding a random code is hard {\em and} distinguishing $\mathcal{D}$ from a random code is hard. However, secure parameters are possible even if the problem to distinguish~$\mathcal{D}$ from a random code turns out to be easy. But, we must be very careful when choosing $\mathcal{C}$ as a way to attack $\miranda$ is to extract from $\mathcal{D}$ the description of $\mathcal{C}$ which enables to decode (which is stronger than distinguishing~$\mathcal{D}$ from a random code). {\it Adding} $\mathcal{A}$ helps to hide its structure, but here we can do more. Once $\mathcal{C}$ has been chosen, we can take a random sub-code $\mathcal{C}_{s}$ of $\mathcal{C}$. The decoding algorithm of~$\mathcal{C}$ is still a decoding algorithm for $\mathcal{C}_{s}$ in its unique decoding regime of parameters: {\it removing} a part of $\mathcal{C}$ by choosing $\mathcal{C}_{s}$ does not change our above discussion but it has the advantage to potentially remove some structures in $\mathcal{C}$ which may be difficult to hide. $\miranda$ thus relies on this ``add and remove'' procedure applied to a decodable code to be as secure as possible.
\medskip




{\bf\noindent Our second contribution: $\miranda$ instantiation with
matrix codes.} Ba\-sed on the previous discussion, to instantiate our
$\miranda$ signature scheme, we reduce to the task of selecting a family of
efficiently decodable codes in their unique decoding regime of
parameters. First, we chose to instantiate $\miranda$ with matrix codes, {\it i.e.,} subspaces of matrices over a finite field. Then we choose as family of decodable codes the family of \emph{matrix Gabidulin codes}~\cite{G85},
that is to say, Gabidulin codes represented in their matrix form.  

Gabidulin codes have been proposed in the past for encryption and were
many times subject to key recovery attacks due to their strong algebraic
structure (arising from their $\mathbb{F}_{q^{m}}$--linearity) that makes them difficult to mask. However, in our
construction (as discussed above) we only rely on choosing a subcode of a matrix Gabidulin
code. In particular, let us emphasize that this subcode is a matrix code that does \emph{not}  
inherit from the $\Fqm$--linear structure of the Gabidulin code. This
remark has been essential to ensure the security of our instantiation.
In short, our masking technique consists of hiding the $\Fqm$--linear
structure of the underlying Gabidulin code. It should be noted that this kind of masking technique has recently been introduced in~\cite{ACDGV24} in the context of encryption.
Still, the masking technique in the present article is different and in particular
is safe with respect to the recent attack proposed in~\cite{PWL25}.

Notice that the principle of masking a linear structure over an
extension is not new. This is precisely the historical core of
multivariate systems, such as Matsumoto Imai's $C^*$ and HFE systems.
On the code side, McEliece scheme~\cite{M78,ABCCGLMMMNPPPSSSTW20} involves Goppa codes which
are designed as subfield subcodes of Reed--Solomon codes over a field
extension. Here again the structure over the field extension is
hidden in the Goppa code's construction.

All in all, we instantiate $\miranda$ with matrix Gabidulin codes and we study its security via a thoroughly algorithmic study. This allows us to offer concrete parameters for $\lambda$ bits of security. For instance, we obtain for $128$ bits of security signature sizes as low as $60$ bytes and a public key size of order $792$ kilobytes. It makes $\miranda$ a competitive scheme compared to $\mathsf{UOV}$ signature schemes (for example $\mathsf{Mayo}$), but also compared to signature schemes derived from the MPC-in-the-head or zero-knowledge proofs paradigms whose drawback is their large signature size. In terms of efficiency $\miranda$ enjoys fast verification but not fast signing time. We provide estimates of the signing time, which is approximately 400 ms for our signature of $60$ bytes (some investigations tend to show that we can easily reach $100$ ms). These estimated times for signatures generation are slower than other signatures, but still practical. Furthermore, $\miranda$ is easily parallelizable (since decoding trials by picking $\vec{a} \in \mathcal{A}$ can be done independently) which means than \texttt{FPGA} version of the system could probably reach signature times in the order of seconds. However such optimized implementations of $\miranda$ are clearly beyond the scope of this work and are left for future works.
\medskip

{\bf\noindent Our third contribution: an encryption scheme relying on the masking.}
We also show how to design an encryption scheme based on the ``add and remove'' masking from $\miranda$. 
We obtain for 128 bits of security ciphertext sizes as low as $64$ bytes for a public key of $85$ kilobytes. It is important to note that the decryption is much faster than that of the signature generation of $\miranda$, since with our construction it only requires a single decoding (roughly speaking we only test one candidate $\vec{a} \in \mathcal{A}$) 
This results in encryption and decryption times of less than 1 ms, similar to those of the encryption scheme of~\cite{ACDGV24}.
\medskip

{\bf\noindent Possible applications of our $\miranda$ signature scheme.}
Despite the large size of the public keys (since they consist of a code indistinguishable from a random code of large length), the scheme has the advantage of producing signatures of very small size. Indeed, those of our Miranda signature are about 20 times smaller than those of MPC-based signatures relying on the decoding a random code problem  in rank metric (see for example RYDE \cite{BCFGJRV25,ABBCDFGJNRTV23}). This very small size makes our signature ideal for use in certain applications 
such as blockchain: all the information stored in the blockchain includes a signature, its size is therefore of critical importance.

On the other hand, as every signature based on the \textit{hash-and-sign} paradigm, our $\miranda$ scheme can be converted into a blind signature. 
\iftoggle{llncs}{}{Such process consists of signing an encrypted message, allowing the signatory to sign the message without knowing the original message. The resulting signature must be verifiable through decryption of the message, just like a standard signature. This can be done using a proof of knowledge based on similar mechanism.} 
We could also consider designing an Identity-Based Encryption (IBE) scheme based on this signature. Indeed, to date, no such secure protocol based on the rank metric exists: if \cite{GHPT17} has already made such a proposal,
it was based on the
\textsf{RankSign} \cite{GRSZ14a} signature scheme which was successfully attacked \cite{DT18b}.

\paragraph{   }
\textit{Missing proofs:} due to page constraints, we were unable to include all the proofs. These proofs will be available in a subsequent online version.


%% file: src/preliminaries.tex
\section{Notation and matrix codes background}

{\bf \noindent Basic notation.} The notation $x \eqdef y$ means that
$x$ is being defined as equal to~$y$. Let $a < b$ be integers, we let
$[a,b]$ denote the set of integers $\{a,a+1,\dots,b\}$. Notation
$\negl(\lambda)$ refers to a function that is $O(1/\lambda^{b})$ for
every constant $b > 0$. Vectors are in row notation and they will be
written with bold letters such as $\vec{a}$. Uppercase bold letters
such as $\vec{A}$ are used to denote matrices. We let~$q$ denote a power of a
prime number and $\mathbb{F}_{q}$ the finite field of
cardinality $q$. Given integers~$m,n$ we denote by
$\mathbb{F}_{q}^{m \times n}$ the set of $m\times n$ matrices with
entries in $\mathbb{F}_{q}$. We denote
the zero matrix in~$\mathbb{F}_{q}^{m \times n}$ by
$\vec{0}_{m\times n}$ or by $\vec 0$ when the matrix size is clear
from the context. Given $\vec{A} \in \mathbb{F}_{q}^{m \times n}$, we let~$\tr \vec{A}$ denote its trace. In what follows, $\mathcal{B}_{t}^{m,n,q}$  ({\it resp.} $\mathcal{S}_{t}^{m,n,q}$),
or simply $\mathcal{B}_{t}$ ({\it  resp.}~$\mathcal{S}_{t}$)
when the ambient space is clear, will
denote the ball  ({\it resp.} sphere)
of radius~$t\leq \min\left( m,n \right)$ around
$\vec{0}_{m\times n}$ in~$\mathbb{F}_{q}^{m\times n}$ for the rank
metric $|\cdot |$ which is defined as
\[\forall \vec{A} \in \mathbb{F}_{q}^{m \times n}, \; |\vec{A}| \eqdef
\rk(\vec{A}) \ .\]
Given vectors $\vv_1, \dots, \vv_s$ in a given space, we let
$\langle \vv_1, \dots, \vv_s\rangle$ denote the subspace they span.

\medskip

{\bf \noindent Matrix codes.} A \emph{matrix code} $\mathcal{C}$ over
$\mathbb{F}_{q}$ with size $m \times n$ and dimension $k$ is a
subspace of dimension $k$ of the vector space
$\mathbb{F}_{q}^{m \times n}$. We say that it has parameters~$[m\times n,k]_{q}$ or that it is an~$[m \times n,k]_{q}$-code.  An
important quantity characterizing a code is its minimum distance. For
a matrix code $\mathcal{C}\subseteq \mathbb{F}_{q}^{m \times n}$, it
is defined as the least rank of its non-zero codewords, {\it
  i.e.,}
\[
  d_{\textup{min}}(\mathcal{C}) \eqdef \min \left\{ |\vec{C}|: \;
    \vec{C} \in \mathcal{C} \setminus \{\vec{0}_{m
      \times n}\} \right\}.
\]
Given an $[m \times n,k]_{q}$-code $\mathcal{C}$, its {\em dual} is defined as 
\[
  \mathcal{C}^{\perp}\eqdef \left\{ \vec{B}\in \mathbb{F}_{q}^{m
      \times n}: \; \forall \vec{C} \in \mathcal{C}, \; \tr\left(
      \vec{C} \vec{B}^{\top} \right) = 0 \right\}.
\]
It defines an $[m \times n, mn-k]_{q}$-code.
\iftoggle{llncs}{\medskip
}{Furthermore, if $\vec{B}_{1},\dots,\vec{B}_{mn-k}$ denotes a basis of $\mathcal{C}^{\perp}$, then 
\[
\mathcal{C}^\perp \eqdef \left\{ \vec{C} \in \mathbb{F}_{q}^{m\times n}: \; \forall i \in [1,mn-k], \; \tr\left( \vec{C}\vec{B}_{i}^{\top} \right) = 0 \right\}.
\]}

{\bf \noindent Probabilistic notation.} For a finite set $\mathcal{E}$, we write $X \Unif \mathcal{E}$ when $X$ is an element of $\mathcal{E}$ drawn uniformly at random.
\iftoggle{llncs}{}{Sometimes, we will use a subscript to stress the random variable specifying the associated
probability space over which the probabilities or expectations are taken. For instance the probability $\mathbb{P}_{X}(E)$ of the event $E$ is taken
over the probability space $\Omega$ over which the random variable~$X$ is defined.}
The {\em statistical distance} between two random variables $X$ and~$Y$ taking their values in a same finite space $\mathcal{E}$ is
\iftoggle{llncs}{}{defined as}
\begin{equation}\label{eq:def1SD} 
	\Delta(X,Y) \eqdef \frac{1}{2} \sum_{a\in \mathcal{E}}
	\left| \mathbb{P}\left(X=a\right) - \mathbb{P}\left( Y= a \right) \right|. 
\end{equation}

%
%
%

%% file: src/signature_via_atpsf.tex
\section{Rank-based signatures via Average Trapdoor Preimage Sampleable Functions}

Our focus in this paper is to design a signature scheme in the
rank-based setting following the \emph{Full Domain Hash} (FDH)
paradigm~\cite{BR96}. Our aim is therefore to design an FDH-scheme
whose security inherits from the (average) hardness of the $\minrank$
problem which is stated in its dual form\iftoggle{llncs}{}{\footnote{$\minrank$ is usually defined in its primal form: given $\vec{M}_1,\dots,\vec{M}_{k} \in \mathbb{F}_{q}^{m \times n}$ and $\vec{Y}=\sum_{i=1}^{k} c_i\vec{M}_i + \vec{X}$ where~$|\vec{X}| \leq t$, find $d_1,\dots,d_k \in \mathbb{F}_{q}$ such that~$\vec{Y}- \sum_{i=1}^{k} d_i \vec{M}_i = \vec{E}$ with $|\vec{E}|\leq t$. The matrices $\vec{M}_{i}$'s are viewed as a basis of some~$[m \times n,k]_{q}$-code $\mathcal{C}$ while in the dual form of $\minrank$, matrices $\vec{B}_{i}$'s are a basis of its dual $\mathcal{C}^{\perp}$.}} as follows. It consists of {\it decoding} a random matrix-code, {\it i.e.,} a matrix-code sampled uniformly at random, given a basis of its dual.  


\begin{definition}[$\minrank$]\label{def:minRank} Let $m,n,k,t,q$ be integers which are functions of some security
  parameter $\lambda$ and such that $mn \geq k$.  Let
  $\vec{B}_1, \dots, \vec{B}_{mn - k} \in \mathbb{F}_{q}^{m \times
    n}$, $\vec{X} \in \mathcal{B}_{t}^{m,n,q}$ be sampled uniformly at
  random and
  \[\vec{s} \eqdef \left( \tr\left( \vec{X}\vec{B}_{i}^{\top}
      \right)\right)_{i = 1}^{ mn-k}.\] The $\minrank(m,n,k,t,q)$
  problem consists\iftoggle{llncs}{, given
  	$(\vec{B}_{1},\dots,\vec{B}_{mn-k},\vec{s})$,}{} of finding~$\vec{E} \in \mathcal{B}_{t}^{m,n,q}$ such that
  $\vec{s} = \left( \tr\left( \vec{E}\vec{B}_{i}^{\top}
    \right)\right)_{i=1}^{mn-k}$\iftoggle{llncs}{.}{ given
    $(\vec{B}_{1},\dots,\vec{B}_{mn-k},\vec{s})$.}
\end{definition}


It leads us to consider the following one-way function, 
\begin{equation}\label{eq:fToInvert}
\fow:\vec{X} \in \mathcal{B}_{t}^{m,n,q} \longmapsto \left( \tr\left( \vec{X}\vec{B}_{i}^{\top} \right)\right)_{i=1}^{mn-k} \in \mathbb{F}_{q}^{mn-k}.
\end{equation} 
The FDH paradigm requires one to be able to invert (thanks to a trapdoor)
this function on any input~$\vec{s} \in \mathbb{F}_{q}^{mn-k}$ (or at least a density close to $1$ of inputs).
  This constrains how
parameters~$m,n,k,t,q$ have to be chosen. In particular the radius~$t$
has to be chosen large enough to ensure that $\sharp \mathcal{B}_t^{m,n,q} \gg q^{mn-k} = \sharp \mathbb{F}_{q}^{mn-k}$, {\it i.e.,} $t$ has to be larger than the so-called Gilbert-Varshamov radius. 
This implies that for any input
$\vec{s}$, we expect an exponential number of preimages. This situation
is reminiscent to the lattice case where FDH schemes as
\textsf{Falcon}~\cite{FHKLPPRSWZ} also rely on inverting a one-way
function with an exponential number of preimages (by considering the
$\mathsf{ISIS}$ problem). The same phenomenon appeared with the
code-based signature \textsf{Wave}~\cite{DST19a}.  This innocent
looking fact has a huge consequence in terms of security: when
inverting 
$\fow$ 
with the help of a
trapdoor we have to be extremely careful on how we choose the inverse,
otherwise we could reveal informations on the trapdoor. Some lattice-
and code-based FDH schemes~\cite{HHPSW03,NCLKK22} were broken~\cite{NR09,DLV24} as they did not take into account this potential
flaw.

Hopefully, building a secure FDH signature in this situation can be
achieved by imposing additional properties to the one-way
function \cite{GPV08}. This is captured by the notion of \emph{Trapdoor Preimage
  Sampleable Functions} (TPSF)~\cite[Def.~5.3.1]{GPV08}. However,
in our case, we will instead rely on a slight variation of TPSF:
\emph{Average Trapdoor Preimage Sampleable Functions} (ATPSF) as introduced
in \cite[Def.~1]{DST19a}. As shown by~\cite{CD20}, it is enough to tightly reduce
the security of our scheme (in the EUF-CMA security model) to the
problem of inverting~$\fow$ without the trapdoor. Roughly speaking, (A)TPSF
are requiring the distribution of the output preimages to be
independent of the trapdoor and thus do not reveal any
information \iftoggle{llncs}{}{about it}. We can define ATPSF in the rank case as follows.


%
\begin{definition}\label{def:rankATPS}
  A \emph{Rank-based Average Trapdoor Preimage Sampleable Functions}
	 is a pair \iftoggle{llncs}{$(\trap, \sampPre)$}{} of probabilistic
	polynomial-time algorithms \iftoggle{llncs}{}{$(\trap, \sampPre)$} together with parameters $m,n,k,t,q$ which are functions of some security parameter $\lambda$. 
	\begin{itemize}\setlength{\itemsep}{5pt}
		\item $\trap$: given $\lambda$, outputs $\left( (\vec{B}_{i})_{i=1}^{mn-k},T\right)$ where $\vec{B}_{1},\dots,\vec{B}_{mn-k} \in \Fq^{m\times n}$ and $T$ the trapdoor
		(corresponding to the matrices $\vec{B}_{i}$'s).

		\item  $\sampPre$: is a probabilistic algorithm which takes as input $T$, $\vec{s} \in \mathbb{F}_{q}^{mn-k}$ and outputs~$\vec{E} \in
		\mathcal{B}_{t}^{m,n,q}$ such that $\left( \tr(\vec{E}\vec{B}_{i}^{\top})\right)_{i=1}^{mn-k}= \vec{s}$.
	\end{itemize}
	Moreover, it satisfies the \emph{average preimage sampling with trapdoor} property~:
	\begin{equation}\label{eq:atpsf}
      \mathbb{E}_{T,(\vec{B}_i)_{i=1}^{mn-k}}\left(\Delta(\sampPre(\vec{s},T),\vec{E})\right)\in\mathsf{negl}(\lambda)
	\end{equation}
    where $\vec{E}\in\mathcal{B}_{t}^{m,n,q}$ and $\vec{s}\in\Fq^{mn-k}$ are uniformly distributed.
\end{definition}
	Given a rank-based ATPSF
$(\trap,\sampPre)$, we easily define a rank-based FDH signature scheme. We first generate the public and secret key as $(\pk,\sk)\eqdef \left( \left( \vec{B}_1,\dots,\vec{B}_{mn-k}\right), T\right)\leftarrow \trap(\lambda)$.
We also select a cryptographic hash function $\hash: \{0,1\}^{*}
\rightarrow \mathbb{F}_{q}^{mn-k}$ and a salt $\texttt{salt}$ 
of size\iftoggle{llncs}{}{\footnote{Adding this salt is crucial for known security reductions in the EUF-CMA security model. We have chosen its size as $\lambda$ to be consistent with the different schemes submitted at the \href{https://csrc.nist.gov/projects/pqc-dig-sig}{NIST signature standardization}.}} $\lambda$.
The signing and verification algorithms \Sgnsk\ and \Vrfypk\ 
are then defined as follows.
\medskip 

\begin{center}
	\begin{tabular}{l@{\hspace{3mm}}|@{\hspace{3mm}}l}
		$\Sgnsk(\vec{m})\!\!: \qquad \qquad \qquad$ & $\Vrfypk(\vec{m}',(\vec{E}',\texttt{salt}'))\!\!:$ \\
		$\quad \texttt{salt} \Unif \{ 0,1 \}^{\lambda}$ &$\quad \vec{s} \leftarrow \hash(\vec{m}' ,\texttt{salt}')$ \\
		$\quad \vec{s} \leftarrow \hash(\vec{m} ,\texttt{salt})$ & \quad \texttt{if} ${\left( \tr( \vec{E}'\vec{B}_{i}^{\top})\right)}_{i=1}^{mn-k} = \vec{s} \texttt{ and } |\vec{E}'| \leq t$\\ 
		$\quad\vec{E} \leftarrow \sampPre(\vec{s},T) $ &  \qquad \texttt{return True}  \\
		$\quad \texttt{return}(\vec{E},\texttt{salt})$& \quad \texttt{else}		\\	
&  \qquad \texttt{return False}  
	\end{tabular} 
\end{center}
\medskip 

In the next section, we instantiate $\trap$ and $\sampPre$ in Algorithms~\ref{algo:trap} and~\ref{algo:sgn}. It will form our proposed signature scheme called $\miranda$. In particular, 
we show in Theorem~\ref{theo:unifPreimage} that the average preimage sampling with trapdoor property \eqref{eq:atpsf} of Definition~\ref{def:rankATPS} holds \iftoggle{llncs}{($\miranda$ signatures are not leaking information).}{.} \iftoggle{llncs}{}{Namely that $\miranda$ signatures do not leak any information on their underlying trapdoor.}

%% file: src/signature_scheme.tex
\section{Miranda signature scheme}

\subsection{Trapdoor: Add-And-Remove matrix codes transformation}

We present in this subsection the associated trapdoor to $\miranda$
(and its rationale). It relies on the following construction.
\begin{definition}[Add-and-Remove matrix-code construction]\label{def:addRemove} Let
  $m, n, k,$ $q,\ell_a, \ell_s$ be integers and~$\mathcal{F}$ be a set of
  $[m \times n,k]_{q}$-codes. The \emph{Add-And-Remove}
  construction
  $\mathsf{AddRemove}\left( \mathcal{F},\ell_a,\ell_s \right)$
  is a family of codes defined as follows. Let~$\mathcal{C} \in \mathcal{F}$ and $\mathcal{C}_{s}$ be an arbitrary subcode
  of codimension $\ell_{s}$. Then, let~$\Amat$ be
  an arbitrary $[m \times n,\ell_a]_{q}$-code such that
  $\mathcal{C} \cap \Amat = \{\vec{0}_{m \times n}\}$. We define
  the resulting \emph{Add-and-Remove} code $\Dmat$ as follows,
\[
\Dmat\eqdef \mathcal{C}_{s} \oplus \Amat \ , 
\]
where $\oplus$ \iftoggle{llncs}{is}{denotes} the direct sum of vector spaces.
It defines an $[m \times n, k-\ell_s+\ell_a]_{q}$-code. 
\end{definition}

\begin{remark}\label{rem:AR}
	Our Add-and-Remove construction can be defined in a much more generality. In particular, it does not require the use of matrix codes.
\end{remark}

$\miranda$ public keys will be defined as a random basis
$\vec{B}_{1},\dots, \vec{B}_{mn -k + \ell_{s}- \ell_{a}}$ of the dual
of some code $\Dmat$ obtained via the $\mathsf{AddRemove}$ construction. Given \iftoggle{llncs}{(the hash of the message to sign)}{}
$\vec{s} \in \Fq^{mn - k +\ell_{s} - \ell_{a}}$ \iftoggle{llncs}{}{(defined as
the hash of the message to sign)}, a signature will be then given by
$\vec{E} \in \mathcal{B}_{t}^{m,n,q}$ such that
\begin{equation}\label{eq:toSolve} 
\left( \tr \left( \vec{E}\vec{B}_{i}^{\top} \right) \right)_{i=1}^{mn-k+\ell_{s} - \ell_{a}} = \vec{s} \ . 
\end{equation} 
The associated trapdoor 
will be the knowledge of the underlying code $\mathcal{C}$ 
to form~$\Dmat$. 
Our rationale behind $\mathsf{AddRemove}$ is
that it ``hides'' the code~$\mathcal{C}$ while its knowledge
enables to solve efficiently the above equation 
  as roughly explain below. Further details are given in
Subsection~\ref{subsec:Invert}. First, notice that, by construction,~$\Dmat^{\perp}$ is contained in~$\mathcal{C}_{s}^{\perp}$ \iftoggle{llncs}{}{(as we
added a code $\Amat$ of dimension~$\ell_a$ to~$\mathcal{C}_{s}$
to form $\Dmat$)} and their dimensions satisfy
\[
\dim \Dmat^{\perp} = mn -  k + \ell_{s} - \ell_a < mn -  k + \ell_{s} = \dim \mathcal{C}_{s}^{\perp} \ . 
\]
We can thus complete any basis of $\Dmat^{\perp}$ with $\ell_{a}$ matrices to form a basis \iftoggle{llncs}{$\vec{B}_{1}$, $\dots$, $\vec{B}_{mn- k + \ell_{s}}$}{$\vec{B}_{1},\dots,\vec{B}_{mn- k + \ell_{s}}$} of~$\mathcal{C}_{s}^{\perp}$. Therefore to solve Equation~\eqref{eq:toSolve} it is enough to 
guess $\vec{t} \in \Fq^{\ell_{a}}$ and to compute $\vec{E} \in \mathcal{B}_{t}^{m,n,q}$ such that 
\begin{equation}\label{eq:toSolve2} 
\left( \tr\left( \vec{E}\vec{B}_{i}^{\top} \right)\right)_{i=1}^{mn-k+\ell_{s}} = (\vec{s},\vec{t}) \ .
\end{equation} 
In other words, we reduced the task of solving a $\minrank$ instance with $mn-k+\ell_{s}-\ell_{a}$ equations as given in~\eqref{eq:toSolve} to guessing a $\vec{t} \in \Fq^{\ell_{a}}$ for which we can solve a $\minrank$ instance but this times with~$mn-k+\ell_s$ equations coming from a basis of the dual of $\mathcal{C}_{s}$! This approach may seem at first sight quite convoluted and useless. However, this reduction 
has the following advantage: we can choose a structured code~$\mathcal{C}_{s}$ for which the associated $\minrank$ instance becomes easy, {\it i.e.,} a code~$\mathcal{C}_{s}$ that we know how to efficiently decode while only a basis of~$\Dmat^{\perp}$ is publicly known. Therefore, under the hypothesis that $\Dmat$ hides the structure of~$\mathcal{C}$, no one can use the knowledge of~$\Dmat$ alone to efficiently solve Equation~\eqref{eq:toSolve}. 

\subsection{$\Fqm$--linear and Gabidulin codes}
All the question now is to choose a family of matrix codes $\mathcal{F}$ admitting an efficient {\it decoding algorithm}, {\it i.e.,}
for which we can easily solve Equation~\eqref{eq:toSolve2}.
Not so many families of such matrix codes are known \cite{SKK10,GMRZ13,ACLN20,CP25}. The most popular family of matrix codes admitting an
efficient decoding algorithm is undoubtedly that of \emph{Gabidulin
  codes} \cite{G85}. These codes belong to a particular sub-class of
matrix codes: $\Fqm$--linear codes. Recall that an
$\Fqm$--linear code~$\mathcal{C}$ with length $n$ and
dimension $\kappa$ is a subspace with $\Fqm$--dimension $\kappa$ of
$\Fqm^{n}$. We say that
it has parameters $[n,\kappa]_{q^{m}}$ or that it is an
$[n,\kappa]_{q^{m}}$-code. The space $\Fqm^n$ is endowed
with the following metric. Given two vectors
$\vec{v}, \vec{w} \in \Fqm^{n}$, their \emph{rank distance} is
\iftoggle{llncs}{}{defined as}
 \[
 |\vec{v} - \vec{w}| \eqdef \dim_{\Fq} \langle v_1-w_1,\dots,v_n-w_n \rangle\ .
\]
The \emph{rank weight} of $\vec{v} \in \Fqm^{n}$ is denoted
$|\vec{v}| \eqdef |\vec{v}-\vec{0}|$. This rank weight of vectors is
directly connected to that of matrices via the following
correspondence.  Choose an $\Fq$--basis
$\mathscr{B} = (b_1, \dots, b_m)$ of $\Fqm$ and define the map
 	\begin{equation}\label{eq:matrix_representation}
 		M_{\mathscr{B}} :  \map{\Fqm^n}{\Fq^{m \times n}}{(v_1,\dots, v_n)}{
 			\begin{pmatrix}
 				v_{11} & \cdots & v_{1n}\\
 				\vdots &        & \vdots \\
 				v_{m1} & \cdots & v_{mn}
 			\end{pmatrix}
 		},
 	\end{equation}
 	where for any $j$, $v_{1j},\dots, v_{mj} \in \Fq$ denote the
 	coefficients of $v_j$  expressed in the basis $\mathscr{B}$. The map~$M_{\mathscr{B}}$
    is an isometry with respect to the rank metric.
    Note that, if $M_{\mathscr{B}}$ depends on the choice of
    the basis $\mathscr B$, the rank $M_{\mathscr{B}}(\vv)$ of a
    vector~$\vv \in \Fqm^n$ does not.

 The set of $\Fqm$--linear codes can thus be viewed as a
 subset of matrix codes having some
 ``extra'' algebraic structure in the same way as, for instance,
 cyclic linear codes can be viewed as ``structured'' versions of linear
 codes. Let $\mathcal{C}$ be an~$\Fqm$--linear code, then, given
 $\vec{c} \in \mathcal{C}$ we have that
 $\alpha \vec{c} \in \mathcal{C}$ for any scalar~$\alpha \in \Fqm$. Therefore $\mathcal{C}$ is left globally invariant by
 the $\Fq$--linear
 mapping~$\vec{x} \in \Fqm^{n} \mapsto \alpha \vec{x}
 \in \Fqm^{n}$. We deduce that
 \iftoggle{llncs}{
 $$
  M_{\mathscr{B}}(\mathcal{C}) \eqdef \left\{
 M_{\mathscr{B}}(\vec{c}): \; \vec{c} \in \mathcal{C} \right\}
 $$}{$M_{\mathscr{B}}(\mathcal{C}) \eqdef \left\{
   M_{\mathscr{B}}(\vec{c}): \; \vec{c} \in \mathcal{C} \right\}$} is
 \iftoggle{llncs}{}{globally} left invariant by the linear mappings
 $M_{\mathscr{B}}(\vec{x}) \in \Fq^{m \times n} \mapsto
 M_{\mathscr{B}}(\alpha\vec{x}) \in \Fq^{m \times n}$ for~$\alpha \in \Fqm$, {\it i.e.,}
 $M_{\mathscr{B}}(\mathcal{C})$ is globally invariant by the left
 multiplication by \iftoggle{llncs}{}{some} matrices~$\vec{P}_{\alpha} \in \Fq^{m \times m}$
  which represent the multiplication-by-$\alpha$ maps in the basis~$\mathscr{B}$.

 
  We are now almost ready to properly define Gabidulin codes. Our last
  ingredient is the set of~$q$-polynomials. A \emph{$q$-polynomial} is a
  polynomial of the form~\iftoggle{llncs}{$P(X) = p_0X + p_1 X^{q} + \dots + p_{d}X^{q^{d}}$ where
  		$p_{i} \in \mathbb{F}_q$ and $p_{d} \neq 0$}{
	\[
      P(X) = p_0X + p_1 X^{q} + \dots + p_{d}X^{q^{d}} \; \mbox{ where
        $p_{i} \in \mathbb{F}_q$ and $p_{d} \neq 0$}\ .
	\]}
    The integer $d$ is called the \emph{$q$-degree} of $P$ and denoted $\deg_q(P)$. 
 
 \begin{definition}[Gabidulin codes]\label{def:gab}
   Let $m,n,\kappa$ be integers where $\kappa \leq n \leq m$ and~$\vec{g} = (g_1,\dots,g_n) \in \Fqm^{n}$ whose
   entries are $\Fq$--linearly independent. The Gabidulin
   code of evaluation vector $\vec{g}$ and~$\Fqm$--dimension $\kappa$
   is defined as
 	\[
      \mathsf{Gab}(\vec{g},\kappa) \eqdef \left\{ \left( P(g_1),\dots, P(g_n)
        \right) ~:~ \deg_q(P) <\kappa \right\}.
 	\]
 	It defines an $[n,\kappa]_{q^{m}}$-linear code with minimum distance
    $n-\kappa+1$.
  \end{definition}

\begin{remark}\label{rem:Gab_above_GV}
  Gabidulin codes are Maximum Rank Distance (MRD) codes, {\it i.e.,}
  they have the largest possible minimum distance (as matrix and
  $\mathbb{F}_{q^m}$--linear codes) for fixed $m, \kappa$. In particular, their
  minimum distance is much greater than that of almost any code, which
  is given by the {\it Gilbert-Varshamov} radius. When $m=n$
  the latter radius is approximately equal to \cite[Example $1$]{L06} \iftoggle{llncs}{$m\left( 1 - \sqrt{\frac{\kappa}{m}} \right)$.}{
	\[
	m\left( 1 - \sqrt{\frac{\kappa}{m}} \right)~.
	\]}
\end{remark}
 
%
%

In the sequel, we will mainly deal with \emph{matrix} Gabidulin
codes. That is to say, the image of a Gabidulin code by some
$M_{\mathscr B}$ map as given in \eqref{eq:matrix_representation}.

\iftoggle{llncs}{}{
\medskip
\noindent \textbf{Sampling matrix Gabidulin codes.}
In the sequel, we will frequently have to sample matrix Gabidulin
codes.  This will be done as follows. Since for our instantiations we
only consider codes of length $m$ \emph{i.e.,} such that $m=n$ we
restrict here to this case.
\begin{enumerate}\setlength{\itemsep}{5pt}
\item Draw a uniformly random $\Fq$--basis $\gv$ of $\Fqm$;
\item Compute a basis $\Gabcode{\kappa}{\gv}$;
\item Draw another uniformly random $\Fq$--basis $\mathscr B$ of
  $\Fqm$;
\item Return a uniformly random $\Fq$--basis of $M_{\mathscr B}(\Gabcode{\kappa}{\gv})$.
 \end{enumerate}
}
\medskip 






\subsection{$\miranda$'s trapdoor}

$\miranda$'s
 $\trap$ algorithm
 basically consists of building a uniform code from
 $\mathsf{AddRemove}\left( \mathcal{F},\ell_a,\ell_s \right)$ where
 $\mathcal{F}$ is the family of matrix Gabidulin codes with parameters
 $[n,\kappa]_{q^{m}}$ viewed as~$[m \times n, \kappa m]_{q}$-matrix
 codes. The trapdoor will be roughly speaking the structure of the
 used Gabidulin code, {\it i.e.,} the basis~$\mathscr B$ chosen for
 the map $M_{\mathscr{B}}$ (see
 Equation~\eqref{eq:matrix_representation}). However, notice that in
 the
 $\mathsf{AddRemove}\left( \mathcal{F},\ell_a,\ell_s
 \right)$-construction, given a code $\mathcal{C} \in \mathcal{F}$
 {\it we choose a
   subcode~$\mathcal{C}_{s}\subseteq \mathcal{C}\subseteq \Fq^{m
     \times n}$}. Here $\mathcal{C}$ will be a Gabidulin code,
 therefore it is~$\Fqm$--linear and it is left globally invariant by
 the left multiplication of matrices which represent the
 multiplication by scalars~$\alpha \in \Fqm$. But by choosing a
 sub-matrix code~$\mathcal{C}_{s} \subseteq \Fq^{m \times n}$ of this
 matrix Gabidulin code, we precisely destroy the $\Fqm$--linearity: a
 sub-matrix code of a matrix Gabidulin code is {\bf not} \emph{\`a
   priori} $\Fqm$--linear. Note that recovering the hidden
 $\Fqm$--linear structure is enough to fully recover the
 secret key.
 
 \begin{algorithm}[htb]
 	\caption{$\trap\left(\lambda\right)$: \textsf{Miranda} trapdoor algorithm outputting $\vec{B}_1,\dots,\vec{B}_{mn-km+\ell_{s}-\ell_{a}}$ and its corresponding trapdoor $T$}\label{algo:trap}\medskip 
 	Parameters: $m,n,\kappa,t,q, \ell_{a},\ell_{s}$ to ensure $\lambda$ bits of security\medskip
 	\hrule
 \setlength{\baselineskip}{1.25\baselineskip}
 	\begin{algorithmic}[1]
 		\State $ \vec{g} \Unif \left\{ (x_1,\dots,x_n) \in \Fqm^{n}\mbox{ whose entries are $\Fq$--linearly independent} \right\}$ 
 		\State $\mathscr{B} \Unif \left\{ \Fq\mbox{--basis of $\Fqm$ as $\mathbb{F}_q$-space}\right\}$
 		\State Compute $\mathcal{C} \eqdef M_{\mathscr{B}}(\mathsf{Gab}(\vec{g}, \kappa))$ \Comment{Matrix Gabidulin code of $\Fq$--dimension $\kappa m$} 
 		\State\label{state:Cs} $\mathcal{C}_{s} \Unif \left\{\text{subcodes of } \mathcal C \text{ of codimension } \ell_{s}\right\}$
 		\State\label{state:Amat} $\Amat \Unif \left\{ \mathcal{U}: \mbox{ $[m\times n,\ell_{a}]_{q}$-codes such that $\mathcal{U} \cap \mathcal{C} = \{ \vec{0}_{m \times n}\}$}  \right\}$
 		\State Compute $\vec{B}_{1},\dots,\vec{B}_{mn-\kappa m  + \ell_{s}-\ell_{a}}$ a {\it random} $\Fq$--basis of $\left( \mathcal{C}_s \oplus \Amat \right)^{\perp}$ 
 		\State Complete the above basis into a basis $\vec{B}_1,\dots, \vec{B}_{mn-\kappa m + \ell_{s}}$ of $\mathcal{C}_{s}^{\perp}$
 		\State $T \leftarrow \left( \vec{g}, \mathscr{B}, \left( \vec{B}_{mn-\kappa m + \ell_{s}-\ell_{a}+1},\dots,\vec{B}_{mn-\kappa m + \ell_{s}}\right)\right)$
 		\State\Return $\left( \left( \vec{B}_{1},\dots,\vec{B}_{mn-\kappa m + \ell_{s}-\ell_{a}}\right),T \right)$ 
 	\end{algorithmic}
 \end{algorithm}

 \begin{remark}
   Instructions \ref{state:Cs} and \ref{state:Amat} of Algorithm~\ref{algo:trap} both consist
   in sampling random subcodes or random codes with trivial intersection with some
   given code $\mathcal{C}$. This task is actually easy to realise since we are working over
   a finite field\iftoggle{llncs}{.}{ and hence only finitely many codes of given length exist.}
 \end{remark}
 
 \iftoggle{llncs}{}{
 \begin{remark}
   Note that the trapdoor could have only consisted in the pair
   $(\gv, \mathscr B)$. Indeed, from $\gv$ and $\mathscr B$ and the
   given basis of $(\C_s \oplus \Amat)^\perp$, it is possible to
   deduce a basis of $\C_s^\perp$.
 \end{remark}

 \begin{remark}
   Notice that in the trapdoor algorithm, {\it i.e.,}
   Algorithm~\ref{algo:trap}, we are returning $\vec{B}_{1}$, $\dots$,
   $\vec{B}_{m(n-\kappa) + \ell_{s}-\ell_{a}}$ matrices corresponding
   to a random dual basis of a code picked uniformly at random among
   the family of
   $\mathsf{AddRemove}\left( \mathcal{F},\ell_a,\ell_s \right)$ codes
   where $\mathcal{F}$ is the set of length $n$ Gabidulin codes with
   $\mathbb{F}_{q^{m}}$--dimension $\kappa$. In particular these
   Gabidulin codes have dimension~$\kappa m$ when viewing them as
   matrix codes.
 \end{remark}
 }
 
\subsection{Miranda: how to invert with the help of the trapdoor?}\label{subsec:Invert}
 
Now, we present how to invert~$\fow$ as defined in Equation \eqref{eq:fToInvert} for the matrices~$\vec{B}_{1},\dots,\vec{B}_{m(n-k) +\ell_{s}-\ell_{a}}$ being output by Algorithm~\ref{algo:trap} and with the help of the associated trapdoor $T$. As mentioned previously, it relies on the decoding a Gabidulin code. \iftoggle{llncs}{}{The first such decoding algorithm was proposed in \cite{G85} and the usual reference on the topic is \cite{L06a}.} In the following proposition, we summarize some of its properties that will be useful for us. However, there exist many optimizations improving its running time, see for instance  \cite[Table A.1]{W13} for a survey of running time with different variants of the algorithm from \cite{G85,L06a}.

 \begin{proposition}[\cite{G85}]\label{propo:decoGab}
 	Given a Gabidulin code $\mathsf{Gab}(\vec{g},\kappa)$ with parameters \iftoggle{llncs}{$m$, $n$, $\kappa$, $q$}{$m,n,\kappa,q$}, {\it i.e.,} given the knowledge of~$\vec{g}\in \Fqm^{n}$ and $\kappa$, there exists a deterministic algorithm $\mathsf{Decode}^{\mathsf{Gab}}$ running in $O(n^{2})$ operations in $\Fqm$ and such that given~$\vec{y} \in \Fqm^{n}$, $\vec{g}$ and $\kappa$, 
 	\begin{itemize}\setlength{\itemsep}{5pt}
 		\item if there exists $\vec{c} \in \mathsf{Gab}(\vec{g},\kappa)$ and $\vec{e} \in \Fqm^n$ with $|\vec{e}| \leq \frac{n-\kappa}{2}$ such that $\vec{y} = \vec{c} + \vec{e}$, then it outputs $\vec{e}$;		
 		\item otherwise, it outputs $\bot$~. 
 	\end{itemize}
 \end{proposition}

 The decoding algorithm of Gabidulin codes enjoys useful properties for our instantiation. First, it is deterministic. But even more importantly, this algorithm finds $\vec{e}$ for any input~$\vec{y} = \vec{c}+\vec{e}$ as soon as $\vec{c} \in \mathsf{Gab}(\vec{g},\kappa)$ and $|\vec{e}| \leq \frac{n-\kappa}{2}$~.
 In other words, the algorithm works \emph{with certainty in the parameter regime of unique decoding}. This property will be crucial to show in Theorem~\ref{theo:unifPreimage} that $\miranda$ signatures do not leak any information on their underlying trapdoor.

 Our algorithm to invert $\fow$ (as defined in Equation
 \eqref{eq:fToInvert}) is given in Algorithm~\ref{algo:sgn}. Before
 providing a rigorous demonstration that it works, let us explain
 intuitively how it works. Our aim is, given
 $\vec{s} \in \Fq^{m(n-\kappa)+\ell_{s}-\ell_{a}}$, to
 find~$\vec{E} \in \mathcal{B}_{t}^{m,n,q}$ such that
 $\left( \tr \left( \vec{E}\vec{B}_{i}^{\top} \right)
 \right)_{i=1}^{m(n-\kappa)+\ell_{s} - \ell_{a}} = \vec{s}$ where the
 the rank of $\Ev$ is at most \iftoggle{llncs}{}{equal to}
  $$
 t \eqdef \frac{n-\kappa}{2}~\cdot
 $$
%
 
 \begin{itemize}\setlength{\itemsep}{5pt}
 	\item[{\bf Step 1.}] First, we sample uniformly at random $\vec{t} \in \Fq^{\ell_a}$ and we compute by linear algebra a $\vec{Y} \in \Fq^{m \times n}$ such that 
 	$$
 	\left( \tr \left( \vec{Y}\vec{B}_{i}^{\top} \right)
 	\right)_{i=1}^{m(n-\kappa)+\ell_{s}} = (\vec{s},\vec{t})~.
 	$$
 	Notice that we used here a basis $\left(\vec{B}_1,\dots,\vec{B}_{m(n-\kappa) + \ell_{s}}\right)$ of $\mathcal{C}_{s}^{\perp}$ by considering the additional~$\ell_{a}$ matrices from the trapdoor. The point of this step is that if we succeed to decode $\vec{Y}$ in $\mathcal{C}_{s}$, {\it i.e.,} if we compute $\vec{E} \in \mathcal{B}_{t}$ such that~$\vec{Y}-\vec{E} \in \mathcal{C}_{s}$, then we deduce that $	\left( \tr \left( (\vec{Y}-\vec{E})\vec{B}_{i}^{\top} \right)
 	\right)_{i=1}^{m(n-\kappa)+\ell_{s}} = \vec{0}$ and therefore our decoder outputs $\vec{E}$ which is a solution we are looking for. The aim of the next steps is precisely to decode $\vec{Y}$ in $\mathcal{C}_{s}$.

 	\item[{\bf Step 2.}] Given $\vec{Y}$, we compute $\vec{y} \eqdef M_{\mathscr{B}}^{-1}(\vec{Y})$ and we try to decode it in the Gabidulin code~$\mathsf{Gab}(\vec{g},\kappa)$. Our rationale is that if it exists $\vec{E} \in \mathcal{B}_{t}$ such that $\vec{Y} - \vec{E} \in \mathcal{C}_{s}$, then we have $\vec{Y} - \vec{E} \in M_{\mathscr{B}}\left(\mathsf{Gab}(\vec{g},\kappa)\right)$ as $\mathcal{C}_{s} \subseteq M_{\mathscr{B}}\left( \mathsf{Gab}(\vec{g},\kappa) \right)$. In other words, in such a case we have $\vec{y} - \vec{e} \in \mathsf{Gab}(\vec{g},\kappa)$ where~$\vec{e} \eqdef M_{\mathscr{B}}^{-1}(\vec{E})$. But~$\vec{e}$ verifies $|\vec{e}| = |\vec{E}| \leq t = (n-\kappa)/2$ as $M_{\mathscr{B}}$ is an isometry and the decoding algorithm $\mathsf{Decode}^{\mathsf{Gab}}$ 
 	will necessarily return $\vec{e}$.

 	\item[{\bf Step 3.}] In the previous step, we decoded $\vec{y}$ in $\mathsf{Gab}(\vec{g},\kappa)$. If the decoding algorithm fails, we return to Step~1 and sample another $\vec{t}$. If the decoding algorithm succeeds, we have computed some~$\vec{e}$ with rank at most $t$ such that~$\vec{y} - \vec{e} \in \mathsf{Gab}(\vec{g},\kappa)$. Therefore,~$\vec{Y} - \vec{E} \in M_{\mathscr{B}}(\mathsf{Gab}(\vec{g},\kappa))$ where~$\vec{E} \eqdef M_{\mathscr{B}}(\vec{e})$ and we are almost done. To conclude, we just need to verify that~$\vec{Y} - \vec{E} \in \mathcal{C}_{s}$ which may not hold as \emph{à priori} only $\mathcal{C}_{s}\subseteq M_{\mathscr{B}}(\mathsf{Gab}(\vec{g},\kappa))$ is supposed. Otherwise, we go back to Step~1 and guess another $\vec{t}$. 
 \end{itemize}

 The whole point of our approach is that given $\vec{s} \in \Fq^{m(n-\kappa) - \ell_{a} + \ell_{s}}$, then for any $\vec{E} \in \mathcal{B}_{t}$ satisfying $\left( \tr \left( \vec{E}\vec{B}_{i}^{\top} \right)
 \right)_{i=1}^{m(n-\kappa)+\ell_{s} - \ell_{a}} = \vec{s}$, there exist exponentially many possible $\vec{t} \in \Fq^{\ell_{a}}$ such that~$\left( \tr \left( \vec{E}\vec{B}_{i}^{\top} \right)
 \right)_{i=1}^{m(n-\kappa)+\ell_{s}} = (\vec{s},\vec{t})$ has a low rank solution $\vec{E}$.
 Therefore, if we guessed a right $\vec{t}$ in Step~{1}, then our algorithm  will return $\vec{E}$ since the decoding algorithm we use for Gabidulin codes is in its unique decoding regime of parameters. More importantly, as we sampled $\vec{t}$ uniformly at random, our inversion algorithm will compute one of the solutions~$\vec{E}$ uniformly at random. This is the key to prove that our
 signature does not leak information\iftoggle{llncs}{.}{ from the secret key.}

 \begin{algorithm}
 	\caption{$\sampPre(\vec{s},T)$: \textsf{Miranda} inversion algorithm outputting an~$\vec{E} \in \mathcal{B}_{\frac{n-\kappa}{2}}^{m,n,q}$
 		 such that $\left( \tr\left( \vec{E}\vec{B}_{i}^{\top} \right)\right)_{i=1}^{m(n-\kappa)+ \ell_{s}-\ell_{a}} = \vec{s}$ where the~$\left( \vec{B}_{i} \right)_{i=1}^{m(n-\kappa)+\ell_{s}-\ell_{a}}$ and~$T=\left( \vec{g}, \mathscr{B},(\vec{B}_i)_{i= m(n-\kappa)+\ell_{s}-\ell_{a}+1}^{m(n-\kappa) +\ell_{s}} \right)$ are output by Algorithm~\ref{algo:trap} \label{algo:sgn}}
 	\begin{algorithmic}[1]
 		
 		\Repeat
 		\State $\vec{t} \Unif \Fq^{\ell_a}$\label{inst:uniformT}
 		\State\label{inst:Y} Compute $\vec{Y} \in \Fq^{m \times n}$: $\left( \tr \left( \vec{Y}\vec{B}_{i}^{\top} \right)
 		\right)_{i=1}^{m(n-\kappa)+\ell_{s}} = (\vec{s},\vec{t}) $
 		\State $\vec{y} \leftarrow M_{\mathscr{B}}^{-1}\left( \vec{Y} \right)$   \Comment{$M_{\mathscr{B}}$ isometry defined in~\eqref{eq:matrix_representation}}
 		\State $\vec{e} \leftarrow \mathsf{Decode}^{\mathsf{Gab}}(\vec{y},\vec{g},\kappa)$ 	\Comment{$\mathsf{Decode}^{\mathsf{Gab}}$ as per Proposition~\ref{propo:decoGab}}
 		\Until $\vec{e} \neq \bot$ \texttt{and} $\left( \tr \left( M_{\mathscr{B}}\left( \vec{e}\right)\vec{B}_{i}^{\top} \right)\label{instruct:Until} 
 		\right)_{i=1}^{m(n-\kappa)+\ell_{s}} = (\vec{s},\vec{t}) $
 		\State \Return $M_{\mathscr{B}}\left( \vec{e} \right)$
 	\end{algorithmic}
  \end{algorithm}
  
 \begin{restatable}{proposition}{unifT}\label{propo:unifSamplingSgn}
   Suppose that Algorithm~\ref{algo:sgn} with input $\vec{s}\in \Fq^{m(n-\kappa) +\ell_s- \ell_{a}}$ returns~$\vec{X} \in \Fq^{m \times n}$, then~$\vec{X}$ is uniformly random in the set:
 	\begin{equation}\label{eq:X_unif}
 	 \left\{ \vec{E} \in \mathcal{B}_{\frac{n-\kappa}{2}}^{m,n,q}: \; \left(\tr\left( \vec{E}\vec{B}_{i}^{\top} \right)\right)_{i=1}^{m(n-\kappa)+\ell_{s}-\ell_{a}} = \vec{s} \right\}.
 	\end{equation}
 \end{restatable}
 \begin{proof}
 	First, notice that if Algorithm~\ref{algo:sgn} outputs $\vec{X}$, then by definition $M_{\mathscr{B}}^{-1}(\vec{X})$ is output by~$\mathsf{Decode}^{\mathsf{Gab}}$ as defined in Proposition~\ref{propo:decoGab} and it is $\neq \bot$ as ensured by Instruction~\ref{instruct:Until} of Algorithm~\ref{algo:sgn}. Therefore, 
 	\begin{equation}\label{eq:cdtPoids} 
 	|\vec{X}| = \left| M_{\mathscr{B}}^{-1}(\vec{X}) \right| \leq \frac{n-\kappa}{2}~\cdot
 	\end{equation} 
 	Furthermore, as also ensured by Instruction~\ref{instruct:Until}, we have $\left( \tr \left( \vec{X}\vec{B}_{i}^{\top} \right) \right)_{i=1}^{m(n-\kappa)+\ell_{s}} = \left( \vec{s},\vec{t} \right)$ for some~$\vec{t} \in \Fq^{\ell_{a}}$. In particular, 
 	\begin{equation}\label{eq:cdtLin}  
 	\left( \tr \left( \vec{X}\vec{B}_{i}^{\top} \right) \right)_{i=1}^{m(n-\kappa)+\ell_{s}-\ell_{a}} = \vec{s} \ .
 	\end{equation} 
 	In other words, Equations \eqref{eq:cdtPoids} and \eqref{eq:cdtLin} show that $\vec{X}$ is sampled in the claimed set (\ref{eq:X_unif}). Let us now show that it is sampled \emph{uniformly at random} among this set. Let~$\vec{E}_1,\dots, \vec{E}_{N} \in \mathcal{B}_{\frac{n-\kappa}{2}}^{m,n,q}$ 
 	be the elements of the set (\ref{eq:X_unif}).  Let~$\vec{t}_1,\dots,\vec{t}_N \in \Fq^{\ell_a}$ be such that 
 	$$
 	\forall j \in [1,N], \quad \left( \tr\left( \vec{E}_{j}\vec{B}_{i}^{\top} \right) \right)_{i=1}^{m(n-\kappa)+\ell_{s}} = (\vec{s},\vec{t}_j)\ .
 	$$
 	Let us show that the $\vec{t}_{j}$'s are all distinct. Let $\vec{t}_{j_0} = \vec{t}_{j_1}$ for some $j_0, j_1 \in [1,N]$. Then we have, 
 	$$
 	{\left( \tr\left( \left( \vec{E}_{j_0} - \vec{E}_{j_1} \right) \vec{B}_{i}^{\top}\right)\right)}_{i=1}^{m(n-\kappa)+\ell_{s}} = \vec{0}\ ,
 	$$
 	{\it i.e.,} $\vec{E}_{j_0} - \vec{E}_{j_1} \in \mathcal{C}_{s}$ as $(\vec{B}_{i})_{i=1}^{m(n-\kappa) + \ell_{s}}$ forms a dual basis of $\mathcal{C}_{s}$ as ensured by Algorithm~\ref{algo:trap}. But by the triangle inequality,
 	$$
 	\left|\vec{E}_{j_0} - \vec{E}_{j_1}\right| \leq 2\frac{n-\kappa}{2} = n-\kappa < n-\kappa+1 \ .
 	$$
    Since the minimum distance of $\mathcal{C}_{s} \subseteq
    M_{\mathscr{B}}\left(
      \mathsf{Gab}(\vec{g},\kappa)\right)$ is greater than
    $n-\kappa+1$ (which is the minimum distance of
    $\mathsf{Gab}(\vec{g},\kappa)$), this proves that $\vec{E}_{j_0} =
    \vec{E}_{j_1}$. Therefore,~$j_0 = j_1$ and all the
    $\vec{t}_j$'s are distinct.
 	
 	Notice now that in Instruction~\ref{inst:uniformT}, if we do not sample a $\vec{t}_{j}$, then there is no chance to pass Instruction~\ref{instruct:Until}. On the other hand, if we sample $\vec{t}_{j}$ for some~$j \in [1,N]$, then the algorithm will output~$\vec{E}_{j}$. Indeed, let $\vec{Y}_{j}$ as computed in Instruction~\ref{inst:Y} for $\vec{t} = \vec{t}_j$, then $\vec{Y}_{j} -\vec{E}_{j} \in \mathcal{C}_{s} \subseteq M_{\mathscr{B}}\left( \mathsf{Gab}(\vec{g},\kappa)\right)$. Therefore, we have $M_{\mathscr{B}}^{-1}\left( \vec{Y} \right) - M_{\mathscr{B}}^{-1}(\vec{E}_j) \in \mathsf{Gab}(\vec{g},\kappa)$ and by Proposition~\ref{propo:decoGab}, given~$M_{\mathscr{B}}^{-1}\left( \vec{Y} \right)$, $\mathsf{Decode}^{\mathsf{Gab}}$ will necessarily output~$M_{\mathscr{B}}^{-1}(\vec{E}_j)$ (it has rank $\leq (n-\kappa)/2$) which will pass Instruction~\ref{instruct:Until}. Therefore, the probability that~$\vec{X} = \vec{E}_{j}$ is equal to the probability that $\vec{t} = \vec{t}_{j}$ in Instruction~\ref{inst:uniformT} which is independent from~$j$ as $\vec{t}$ is sampled uniformly at random. This concludes the proof. \iftoggle{llncs}{\qed}{}
 \end{proof}
  \iftoggle{llncs}{}{
 \begin{remark}\label{rem:unif}
   Notice that in the proof of the above proposition we only used the
   fact that we can decode a Gabidulin code in its unique decoding
   regime of parameters, {\it i.e.,} for noisy codewords where the
   error's rank is less than half the minimum distance of the underlying code.
 \end{remark}
}

 Up to now, we have proved that Algorithm~\ref{algo:sgn} is correct: any of its outputs $\vec{E}$ has rank $t \leq (n-\kappa)/2$ and verifies
$   \left( \tr\left( \vec{E}\vec{B}_{i}^{\top} \right) \right)_{i=1}^{m(n-\kappa) + \ell_{s} - \ell_{a}} = \vec{s}$ where $\vec{s}$ is the hash of the message to be signed.
Moreover, the output is uniformly random among all the possible solutions $\vec{E}$.
However we did not give its average running time.
All the difficulty is that given some~$\vec{s} \in \Fq^{m(n-\kappa) + \ell_{s}-\ell_{a}}$, we are looking for some $\vec{t} \in \Fq^{\ell_a}$ such that there exists $\vec{E} \in \mathcal{B}_{t}$ verifying
$$
\left( \tr\left( \vec{E}\vec{B}_{i}^{\top} \right) \right)_{i=1}^{m(n-\kappa) + \ell_{s} -\ell_a} = \left( \vec{s},\vec{t} \right). 
$$
In the following proposition we give (in average over the \iftoggle{llncs}{}{algorithm} input $\vec{s}$) the average number of guesses of $\vec{t}$ before finding one for which the algorithm terminates.


\begin{remark}
	We could add a list to the algorithm enumerating the $\vec{t}$ to return $\bot$ once all $\vec{t}$ were picked. This is not necessary because, the probability of picking an entry with no solutions is negligible with our parameter choice (see Remark~\ref{rem:sgn}).
\end{remark}

\begin{proposition}\label{propo:runningTime}
  Let $\vec{s} \Unif \Fq^{m(n-\kappa) + \ell_{s}-\ell_{a}}$ be an input of Algorithm~\ref{algo:sgn} and we let $T_{\vec{s}}$ denote its average running-time over its internal randomness. We have
	$$
	\mathbb{E}_{\vec{s}}\left( 1/T_{\vec{s}} \right) = \Theta\left( \frac{1}{m^{3}n^{3}} \cdot \frac{\sharp \mathcal{B}_{\frac{n-\kappa}{2}}^{m,n,q} }{q^{m(n-\kappa) +\ell_{s}}} \right) = \Theta\left(  \frac{1}{m^{3}n^{3}} \cdot q^{  \frac{\left( n-\kappa\right) \cdot \left( m + n  + \kappa\right)}{4} - \ell_{s} } \right) ~\cdot
	$$
\end{proposition}

 The above proposition shows how to chose parameters such that the average cost of Algorithm~\ref{algo:sgn}, {\it i.e.,} the number of guesses of $\vec{t}$ in Instruction~\ref{inst:uniformT}, is small
 \iftoggle{llncs}{(parameters in Section~\ref{sec:param} will ensure that it will  range between $2^{4}$ and $2^{16}$)}{(our parameter choice in Section~\ref{sec:param} will ensure that it will  not exceed $2^{16}$ and it will typically ranging between $2^{4}$ and $2^{16}$).} 
 \newline

 {\bf \noindent No leakage of $\miranda$ signatures.} We have instantiated $\sampPre$ associated to $\miranda$ in Algorithm~\ref{algo:sgn} and studied its average running time. However, we have not proved that this algorithm does not leak any information about the trapdoor $T$ given as input, namely that $\miranda$ satisfies Property~\eqref{eq:atpsf} of Definition~\ref{def:rankATPS}. Our aim in what follows is to prove it. This will prove that $\miranda$ is a rank-based average trapdoor preimage sampleable function (see Definition~\ref{def:rankATPS}).

 It turns out that we already have at our disposal the key-ingredient to prove that $\miranda$ signatures do not leak information on their underlying trapdoor: for a fixed input $\vec{s}$, Algorithm~\ref{algo:sgn} either outputs one of the possible signatures~$\vec{E}$ uniformly at random, or it fails (after testing all the~$\vec{t} \in \Fq^{\ell_a}$ in Instruction~\ref{inst:uniformT}) as shown by Proposition~\ref{propo:unifSamplingSgn}. This is indeed almost enough to prove that $\miranda$ verifies Equation \eqref{eq:atpsf}. The following generic lemma shows that we only need to prove that for~$\vec{E} \Unif \mathcal{B}_{t}$, the distribution of $\fow(\vec{E})$ (see Equation~\eqref{eq:fToInvert}) is close to the uniform distribution on~$\Fq^{m(n-\kappa) +\ell_{s}-\ell_{a}}$ with respect to the statistical distance.
  Indeed, 
  the statistical distance between~$\vec{E}_{\vec{s}}$ (notation `$x$' in Lemma~\ref{lemma:invertF} below) being an output of Algorithm~\ref{algo:sgn}, {\it i.e.,} an inverse~$\fow^{-1}(\vec{s})$, for input~$\vec{s}\Unif \Fq^{m(n-\kappa)+\ell_{s}-\ell_{a}}$ (notation `$y$' in the lemma) and~$\vec{E} \Unif \mathcal{B}_{t}$ (notation `$x_u$' in Lemma~\ref{lemma:invertF} below) is equal to the statistical distance between $\vec{s}$ and~$\fow(\vec{E})$.


 \begin{lemma}\label{lemma:invertF}
 	Let $f: \mathcal{E} \rightarrow \mathcal{F}$ where $\mathcal{E},\mathcal{F}$ are some finite sets. Let, $y \Unif \mathcal{F}$ and~$x_u\Unif \mathcal{E}$. Let~$x \Unif f^{-1}(y)$ but in the case where $f^{-1}(y) = \emptyset$ then by convention $x = \bot \notin \mathcal{E}$. 
 	We have 
 	$$
 	\Delta(x,x_u) =\Delta\left( y, f(x_u) \right).
 	$$
 \end{lemma}

\begin{remark}\label{rem:y}
	 Note that since $y$ is uniformly distributed, $\mathbb{P}(x = \bot) = \#\{y_0\in \mathcal{F}:\;  f^{-1}(y_0) = \emptyset \}/ \#\mathcal{F}$. But (as $x_u$ cannot be $\bot$) $\Pr(x = \bot)/2 \leq \Delta(x,x_u) = \Delta(y,f(x_u))$. We deduce that the probability to pick $y_0$ such that $f^{-1}(y_0) = \emptyset$ is smaller than $2\Delta(y,f(x_u))$.  
%
\end{remark}

The above lemma shows that $\miranda$ signatures do not leak any information, namely that Equation \eqref{eq:atpsf} is verified, under the condition that, 
$$
\mathbb{E}_{T,{(\Bv_i)}_i}
\left( \Delta\left( \fow(\vec{E}),\vec{s} \right) \right) \in \negl(\lambda)
$$ 
where $\vec{E} \Unif \mathcal{B}_{\frac{n-\kappa}{2}}$, $\vec{s} \Unif \Fq^{m(n-\kappa)+\ell_{s}-\ell_{a}}$ and the $\vec{B}_{i}$'s being output by Algorithm~\ref{algo:trap}.
It turns out that proving this can be done quite easily via the so-called {\em left-over hash lemma} over the functions~$\fow$\iftoggle{llncs}{.}{ which are implicitly a family of (hash) functions  indexed by the matrices $\vec{B}_{i}$'s.}
\begin{proposition}
  \label{proposition:domainSampleability}
	Denote by
	$\left(\vec{B}_{i}\right)_{i=1}^{m(n-\kappa)+\ell_{s}-\ell_{a}}$ the outputs of Algorithm~\ref{algo:trap}. Let $\vec{E} \Unif \mathcal{B}_{\frac{n-\kappa}{2}}^{m,n,q}$ and
	$\vec{s} \Unif \Fq^{m(n-\kappa) + \ell_{s}-\ell_{a}}$,
    we have,
	\[
	\mathbb{E}_{T,\left(
		\vec{B}_{i}\right)_{i=1}^{m(n-\kappa)+\ell_{s}-\ell_{a}}}\left(\Delta\left( \left( \tr
	\left(\vec{E}\vec{B}^{\top}_{i}\right)\right)_{i=1}^{m(n-\kappa) + \ell_{s}-\ell_{a}},\vec{s}
	\right)\right) = O\left( \sqrt{\frac{q^{m(n-\kappa)+\ell_{s}-\ell_{a}}}{\sharp \mathcal{B}_{\frac{n-\kappa}{2}}^{m,n,q}}}
	\right)~.\]
\end{proposition}

  \iftoggle{llncs}{}{
\begin{remark}\label{rem:lhl}
	The proof of Proposition \ref{proposition:domainSampleability} crucially relies on the fact that we have instantiated $\miranda$ via the $\mathsf{AddRemove}\left( \mathcal{F},\ell_a,\ell_s \right)$-construction for a family of codes $\mathcal{F}$ with minimum distance~$\geq n-\kappa+1$. In particular, we could have stated this proposition more generally with $\mathcal{F}$ being a family of codes with minimum distance $\geq d$, not necessarily Gabidulin codes which are implicitly used in the above proposition (they are used to build the $\vec{B}_{i}$'s output by Algorithm~\ref{algo:trap}). 
  \end{remark}
}

\begin{remark}\label{rem:sgn}
	In Section~\ref{sec:param} we choose parameters such that the above bound is negligible. Combining this with Remark~\ref{rem:y} and Proposition~\ref{propo:unifSamplingSgn} shows that the probability to pick an an unsignable $\vec{s}$ is negligible. 
\end{remark}

Combining Propositions~\ref{propo:unifSamplingSgn} and~\ref{proposition:domainSampleability}  with Lemma~\ref{lemma:invertF} shows that $\miranda$ signatures do not leak any information on their underlying trapdoor.

\begin{theorem}\label{theo:unifPreimage}
	Let ${\left(\vec{B}_{i}\right)}_{i=1}^{m(n-\kappa)+\ell_{s}-\ell_{a}}$ as output by Algorithm~\ref{algo:trap} and $\vec{E}_{\vec{s}} \in \mathcal{B}_{\frac{n-\kappa}{2}}^{m,n,q} \cup \{ \bot\}$ be the output of Algorithm~\ref{algo:sgn} with input $\vec{s} \Unif \Fq^{m(n-\kappa)+\ell_{s}-\ell_{a}}$. Let $\vec{E} \Unif \mathcal{B}_{\frac{n-\kappa}{2}}^{m,n,q}$\iftoggle{llncs}{,}{. We have,} 
	\begin{equation*}
		\mathbb{E}_{T,\left(
			\vec{B}_{i}\right)_{i=1}^{m(n-\kappa)+\ell_{s}-\ell_{a}}}\left(\Delta
		\left(\vec{E}_{\vec{s}},\vec{E}\right) \right) = O\left( \sqrt{\frac{q^{m(n-\kappa)+\ell_{s}-\ell_{a}}}{\sharp \mathcal{B}_{\frac{n-\kappa}{2}}^{m,n,q}}}
		\right)~\cdot
	\end{equation*}
\end{theorem}

 \subsection{About the genericity of our approach}

 We have fully instantiated $\miranda$ and shown that it is a rank-based Average Trapdoor Preimage Sampleable Function (ATPSF). It can thus be used as a Full-Domain-Hash function whose security inherits from $\minrank$ hardness, which is a quantum resistant assumption. It is worth noting that very few FDH schemes which are quantumly secure are currently known: $\mathsf{Wave}$ \cite{DST19a} in the code-based setting, $\mathsf{Falcon}$~\cite{FHKLPPRSWZ}  in the lattice-based  setting or $\mathsf{UOV}$-based signatures like $\mathsf{Mayo}$~\cite{mayo}. $\miranda$ has therefore been added to this short (non exhaustive) list of quantum resistant FDH schemes. However, it turns out that we have not just instantiated a new secure FDH-signature relying on the hardness of $\minrank$: we have also introduced \emph{a new framework for designing secure FDH signatures}.

 The key ingredient for instantiating $\miranda$ as an ATPSF has been the $\mathsf{AddRemove}\left( \mathcal{F},\ell_a,\ell_s \right)$-construction with
 $\mathcal{F}$ being the family of Gabidulin codes. However, so far (to prove that signatures do not leak information), we have only used one fact: we can deterministically decode Gabidulin codes in their unique decoding regime of parameters \iftoggle{llncs}{}{(see Remarks~\ref{rem:unif} and~\ref{rem:lhl})}. This quite innocent looking remark shows that $\miranda$ can be instantiated  with \emph{any} matrix code admitting an efficient decoder in its unique decoding regime of parameters, showing that our scheme is not only about Gabidulin (matrix) codes.

 The situation is even more general than with matrix codes equipped with rank metric (see Remark~\ref{rem:AR}). Our approach also works with linear codes equipped for instance with the Hamming metric or lattices with the Euclidean metric. The only needed concept is the notion of \emph{unique decoding}. In some way, our framework is like McEliece's framework~\cite{M78} in the meaning that if one wishes to instantiate McEliece's encryption, one has just to choose a family of codes which can be efficiently decoded. $\miranda$ offers the same degree of freedom: if one wishes to instantiate an FDH signature scheme with no leaking signatures, one has just to choose a family of codes which can be efficiently decoded in their unique decoding regime of parameters. However such as for McEliece, one needs to be \emph{extremely} careful on the choice of the family of codes when instantiating $\miranda$. Historically, for McEliece, many families of codes equipped with an efficient decoder were proposed, and most of them lead to devastating attacks. Only very few families of codes are still considered as secure in McEliece's framework: $\textsf{MDPC}$ codes~\cite{AABBBBDGGGGMPRSTVZ22} and Goppa codes~\cite{G70,G71} as in McEliece's original instantiation~\cite{M78}. The issue is the same with $\miranda$: when using the $\mathsf{AddRemove}\left( \mathcal{F},\ell_a,\ell_s \right)$-construction we have to be \emph{extremely} careful on how we set $\mathcal{F}$. Let us see \iftoggle{llncs}{}{via an example} how a bad choice of code family leads to an attack.

 A natural candidate to instantiate $\miranda$ is the family of
 \emph{Reed-Solomon} codes. Recall that they are subspaces of
 $\Fq^{n}$ be defined as (for
 $\vec{x} = (x_1,\dots,x_n) \in \Fq^{n}$ where all the
 $x_i$ are different) \iftoggle{llncs}{$\mathsf{RS}(\vec{x},\kappa) \eqdef \left\{ \left( P(x_1),\dots,P(x_n) \right): \; P \in \Fq[X]_{<\kappa} \right\}$}{
 $$
 \mathsf{RS}(\vec{x},\kappa) \eqdef \left\{ \left( P(x_1),\dots,P(x_n) \right): \; P \in \Fq[X]_{<\kappa} \right\}~.
 $$}
 These codes have dimension $\kappa$ and they can be decoded at Hamming distance~$\leq \frac{n-\kappa}{2}$ which corresponds to their unique decoding regime of parameters. Unfortunately Reed-Solomon codes are \iftoggle{llncs}{non-secure}{a non-secure choice} to instantiate $\miranda$ when~$\kappa$ is not too large. 
 Indeed, with such instantiation our public-key is a basis (or a dual basis \iftoggle{llncs}{}{but it amounts to the same} as from a basis we easily compute a dual basis and conversely) of a code\iftoggle{llncs}{ $\Dmat \eqdef \mathcal{C}_{s}\oplus \Amat$ where $\mathcal{C}_{s} \subseteq \mathsf{RS}(\vec{x},\kappa)$ and $\mathcal{C}_{s} \cap \Amat = \{\vec{0}\}$}{
 $$
\Dmat \eqdef \mathcal{C}_{s}\oplus \Amat \quad \mbox{ where } \mathcal{C}_{s} \subseteq \mathsf{RS}(\vec{x},\kappa) \mbox{and } \mathcal{C}_{s} \cap \Amat = \{\vec{0}\}~.
 $$} 
 Beware that it was recently shown \cite{CPTZ25} that recovering the hidden Reed--Solomon code from~$\Dmat$ is easy as long as the codimension of $\C_s$ is low.
 
 The above example of a non-secure instantiation of $\miranda$ may have puzzled some readers. Indeed we were just claiming that $\miranda$ is a secure FDH-scheme. There are no contradiction here. What we have shown up to now is that $\miranda$ offers a framework to easily instantiate an FDH-scheme whose signatures are not leaking any information about their trapdoor (via ATPSF). This enables to show (in the EUF-CMA security-model) that breaking $\miranda$ amounts to invert the public trapdoor-one-way function (without knowing its associated trapdoor) which in this case amounts to decode a code from the $\mathsf{AddRemove}$-construction by only knowing a random basis of it. In our above example we have discussed a non-secure instantiation of the 
 $\mathsf{AddRemove}$-construction with Reed-Solomon codes. The goal
 of the next section is precisely to study the security of this
 instantiation of the $\mathsf{AddRemove}$-construction with Gabidulin
 codes.

%% file: src/security.tex
\section{Security}\label{sec:secu}
In this section, we discuss the security of $\miranda$ instantiated
with Gabidulin codes.  Based on Definition~\ref{def:rankATPS}, the
security (in the EUF-CMA security model) of our signature rests on the
hardness of inverting function $\fow$ in
Equation~\eqref{eq:fToInvert}, where the $\Bv_i$'s are sampled by the
trapdoor algorithm. Equivalently, the security reduces to the
following problem.

\begin{definition}[$\mathsf{MinAddRemove}^{\mathsf{Gab}}$]\label{pb}
 Given parameters $m,n,\kappa,q,\ell_a, \ell_s$, let $\Dmat$ be a uniformly random $\addremove$ matrix Gabidulin code
  (see Definition~\ref{def:addRemove} where $\mathcal{F}$ being the family of Gabidulin codes with parameters $m,n,\kappa$ as per Definition~\ref{def:gab}). Given a random basis $\Bv_1, \dots, \Bv_{m(n-\kappa)+\ell_s - \ell_a}$  of its
  dual and a uniformly random $\sv \in\mathbb{F}_{q}^{m(n-\kappa) + \ell_{s} - \ell_a}$,
  find $\Ev$ of rank $t = \frac{n-\kappa}{2}$ such that
  \[    
    {\left(\tr(\Bv_i \Ev^\top)\right)}_{i=1}^{m(n-\kappa)\ell_{s}-\ell_a} = \sv\ .
  \]
\end{definition}

In the sequel, we discuss the main attacks we identified on this problem.
As a decoding problem, one can consider two general approaches:
\begin{itemize}\setlength{\itemsep}{5pt}
\item either applying a generic $\minrank$ solver;
\item or performing a key-recovery attack.
\end{itemize}
For the use of a $\minrank$ solver, we consider attacks proposed by~\cite{GC00,BBCGPSTV20,BBBGT23}.
We give more details in Section \ref{sec:param}.

The best complexity we found was based on a key recovery attack that we
describe in the sequel.  We will first observe in
Section~\ref{ss:distinguisher} that \addremove Gabidulin codes are
efficiently distinguishable from random  with our parameter choices though the distinguisher
we found is by nature exponential. Hence the security of
$\miranda$ does not
reduce with our parameter choice to the hardness of $\minrank$. However, the
difficulty of exploiting such a distinguisher is discussed in the sequel and we claim that
it cannot be used to solve $\mathsf{MinAddRemove}^{\mathsf{Gab}}$ problem (as per Definition~\ref{pb}).

Finally, the best attack we found and which we describe in what follows
essentially consists of recovering the hidden $\Fqm$--linearity of the Gabidulin
code before the \addremove operation.

\subsection{Computing low rank elements in matrix codes}\label{ss:basic_routine}
The distinguisher and the attack to follow both use as a basic routine
the search of low rank codewords in matrix codes. This can be
performed using the so-called information set decoding-like algorithms
such as in the Hamming setting. We refer the readers to
\cite{OJ02,GRS16,AGHT18} for references on those algorithms. Since we
are considering matrix codes that are not $\Fqm$--linear, our task is
slightly different and hence, we will describe the calculation of
low rank codewords in a self contained manner.

Let $\Cmat \subseteq \Fq^{m \times n}$ be a matrix code of dimension $k$.
Denote $a,b$ the positive integers such that
\begin{equation}\label{eq:decomp_K}
  k = am + b \qquad \text{where} \qquad 0 < b \leq m\ .
\end{equation}
Note first that, by Gaussian elimination, one can construct a nonzero matrix
of $\Cmat$ whose $a$ leftmost columns are zero and whose $b-1$ top entries of
the $(a+1)$--th column are zero too. That is to say a matrix of the shape given in
Figure~\ref{fig:matrix_with_man_zeroes}.

\begin{figure}[!h]
  \centering
    \begin{tikzpicture}
      \draw (0,0) rectangle (3,3);
      \draw[<->] (-.2,0) -- (-.2,3);
      \draw[<->] (0,3.2) -- (3,3.2);
      \draw[<->] (0,-.2) -- (1,-.2);
      \node[below] at (.5,-.2) {$a$};
      \draw[-] (1,0) -- (1,3);
      \draw[<->] (1.45,3) -- (1.45,1.6);
      \node[right] at (1.45,2.3) {${b-1}$};
      \draw[-] (1.3,3) -- (1.3,0);
      \draw[-] (1,1.6) -- (1.3,1.6);
      \node[left] at (-.2,1.5) {$m$};
      \node[above] at (1.5,3.2) {$n$};
      \node at (.5,1.3) {${(0)}$};
      \node at (1.15,2.8) {$\scriptstyle 0$};
      \node at (1.15,1.8) {$\scriptstyle 0$};
      \node at (1.15,2.4) {$\scriptstyle \vdots$};
      \node at (1.15,1.4) {$\scriptstyle *$};
      \node at (1.15,.2) {$\scriptstyle *$};
      \node at (1.15,.9) {$\scriptstyle \vdots$};
      \node at (2,1.3) {${(*)}$};
    \end{tikzpicture}
      \caption{Element of an $m \times n$ matrix code of dimension
    $k = am+b$ obtained from Gaussian elimination in the matrix code.}
    \label{fig:matrix_with_man_zeroes}
\end{figure}
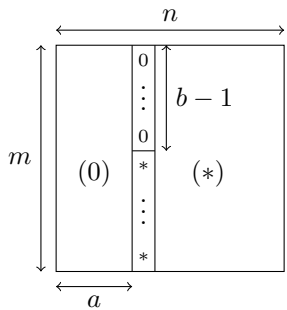
Indeed, imposing such entries to be zero corresponds to $k-1$ linear
constraints, hence there is a nonzero solution $\Cv$ in $\Cmat$. This
matrix $\Cv$ has $a$ zero columns and hence rank at most
$n-a$. The key of the algorithm is that we can iterate this process
while changing the linear constraints until we obtain a matrix $\Cv$
of rank $s < n-a$ where $s$ is some fixed target. This can be
performed by repeating Algorithm~\ref{algo:low_weight}.
\begin{algorithm}
   	\caption{Computing codewords of fixed rank}\label{algo:low_weight}\medskip 
 	\textbf{Inputs.} Code $\Cmat$, target rank $s \leq n-a$.\\
    \textbf{Outputs.} $\Cv \in \Cmat$ of rank $s$ if exists.\medskip
 	\hrule
 \setlength{\baselineskip}{1.25\baselineskip}
 \begin{algorithmic}[1]
   \State {Pick a uniformly random matrix} $\Qv \in \mathbf{GL}_n(\Fq)$ {and compute}
   $\Cmat \Qv$
    \State {Compute a nonzero matrix} $\Cv \in \Cmat \Qv$ \texttt{whose} $a$
    {leftmost columns are zero and} $b-1$ {top entries of the} $(a+1)${-th
    column are zero either (as in Figure~\ref{fig:matrix_with_man_zeroes})}
    \If{$|\Cv| \leq s$}
    \State \Return $\Cv \Qv^{-1}$;
    \Else{} {{go to} $1$.}
    \EndIf
 \end{algorithmic} 
\end{algorithm}

\iftoggle{llncs}{}{
\begin{remark}
  Note that if $b = 1$, then we are looking for matrices with
  prescribed zero columns and, when applied to a $\Cmat \Qv$, the process
  is nothing but the so-called \emph{shortening} of the code with respect to a
  given subspace of $\Fq^n$ as introduced in \cite{BR17,S19a}.
\end{remark}
}

The remaining question is \emph{how many times shall we iterate?} When $\Cmat$
is a random code, then $\Cv$ is a uniformly random matrix with prescribed
zeroes. The following statement yields the probability that such a matrix has rank
$s < m-a$.

Recall that $\mathcal{S}_s^{m,n,q}$ denotes the rank metric sphere of radius $s$, \emph{i.e.}
the set of $m\times n$ matrices over $\Fq$ with rank $s$.
In the sequel we denote
\begin{equation}\label{eq:sigma}
  \sigma_s^{m,n,q} \eqdef \sharp \mathcal{S}_s^{m,n,q}
\end{equation}
and recall that (see \cite{L06})
\begin{equation}\label{eq:approx_sigma}
  \sigma_s^{m,n,q} \approx q^{s(m+n-s)}.
\end{equation}

\begin{proposition}\label{prop:prob_rank}
  Suppose that $\Cmat\subseteq \Fq^{m \times n}$ is a uniformly random
  code of dimension $k = am+b$ for some $a\geq 0$ and
  $1 \leq b \leq m$. Let $s$ be a positive integer\iftoggle{llncs}{,}{ such that,}
  \[
    s < n-a.
  \]
  Let $\Cv \in \Cmat$ be a non-zero matrix whose $a$ leftmost columns
  are zero and the $b-1$ topmost entries of the $(a+1)$--th column are
  zero either (\emph{i.e.,} with the same profile as in
  Figure~\ref{fig:matrix_with_man_zeroes}). Then, the probability that
  $\Cv$ has rank $s$ equals
    \begin{align*}
    \sigma_{s}^{m,n-a-1,q} q^{m (a-n) + b-1} & +
    \frac{\sigma_s^{m, n-a,q} - \sigma_s^{m,n-a-1,q}}{(q^m-1)q^{m(n-a)}} (1-q^{b-1-m})
    \end{align*}
    which is $\approx q^{s(m+n-a-s)+m(a-n)+\max(0,b-1-s)}.$
\end{proposition}

\begin{proposition}\label{prop:complexity_finding_low_weight}
  The cost of a single loop of Algorithm~\ref{algo:low_weight} costs
  $O(km^\omega + mn k^{\omega - 1})$ operations in $\Fq$, where
  $\omega$ denotes the complexity exponent of linear algebra. The average complexity of
    Algorithm~\ref{algo:low_weight} is of
\[
  O \left((km^\omega+mnk^{\omega-1})q^{m(n-a) - s(m+n-a-s) + \min (0,
      s-b+1)} \right)\ \text{operations\ in}\ \Fq \ .
\]
\end{proposition}

\begin{proof}
  Computing a basis of $\Cmat \Qv$ consists of performing $k$ products
  of $m \times m$ matrices, hence a cost in $O(km^\omega)$ where
  $\omega$. Then, the calculation of $\Cv$ is done using Gaussian
  elimination: regarding an $m \times n$ matrix as a vector of length
  $mn$ we have to perform elimination on a $k \times mn$ matrix
  which costs $O(mn k^{\omega-1})$ operations. This yields the
  complexity of a single loop.
  The average number of loops
  is given by Proposition~\ref{prop:prob_rank}. Hence the result.
  \iftoggle{llncs}{\qed}{}
\end{proof}


\subsection{An exponential time distinguisher based on (the lack of) low rank
  codewords}\label{ss:distinguisher}
Our public key is the code
\iftoggle{llncs}{
\(
  \Dmat = \C_s \oplus \Amat,
\)
  }{
\[
  \Dmat = \C_s \oplus \Amat,
\]
}
where $\C_s$ is a sub-matrix code of codimension $\ell_s$ of a matrix
Gabidulin code of $\Fqm$--dimension $\kappa$ and hence of $\Fq$--dimension
\[
  k \eqdef \kappa m
\] and $\Amat$ is a random code such that
$\dim \Amat = \ell_{a}$ and $\C \cap \Amat = \{\vec 0\}$.

For the sake of simplicity, for the description of the distinguisher
we assume $\ell_s$ to be zero. This is the case for some proposed
parameters and even if it does not hold for other ones, the integer~$\ell_s$ remains small and hence will not harm that much the
efficiency of our proposed distinguisher.

Our distinguisher works as follows. Since $\Dmat$ contains a matrix Gabidulin
code, then $\Dmat^\perp$ is a subcode of codimension $\ell_a$ in an
$[m\times n, r m]$ Gabidulin code where
\[r \eqdef n-\kappa~.\] Therefore, the minimum distance of
$\Dmat^\perp$ is at least $m-r+1$ and hence, $\Dmat$ cannot contain
matrices of rank $n-r$ while, as already mentioned in
Remark~\ref{rem:Gab_above_GV}, random codes of the same dimension
do. Let,
\[
  \ell_a = \lambda m + \mu \quad \text{for} \quad
  0 \leq \mu < m~.
\]
Therefore,
\begin{equation}\label{lambda_mu}
  \dim \Dmat^\perp = (r-\lambda-1)m + (m-\mu)\quad  \text{with}\quad
  0 < m-\mu \leq m \ .
\end{equation}
As already mentioned, in a random code of dimension $k = mr-\ell_a$,
according to Proposition~\ref{prop:complexity_finding_low_weight} applied to~$a = r-\lambda-1$,
$b = m-\mu$ and $s = n-r$,  codewords of rank $s = n-r$ can be found in
\iftoggle{llncs}{
\(
  \widetilde{O}\left(q^{r(\lambda+1 + m-n) + \min (0, \mu-r+1)} \right)\)
  operations in $\Fq$.
}{
\[
  \widetilde{O}\left(q^{r(\lambda+1 + m-n) + \min (0, \mu-r+1)} \right)\quad
  \text{operations\ in}\ \Fq \ .
\]
}
Therefore, to distinguish between $\Dmat$ and a random code, it suffices to seek for codewords with rank-weight $\kappa = n-r$ which has the aforementioned cost. If we find such codeword, then the code is random, otherwise it is likely to come from the $\mathsf{AddRemove}$ construction applied to a matrix Gabidulin code. 

Notice that our distinguisher has exponential time complexity. However, its running time will far be below the security level for the parameters we choose in Section~\ref{sec:param} ($r$ is chosen as being quite small).

\medskip 

{\bf \noindent Comment about this distinguisher.}
We gave a distinguisher on the public keys, which despite its
exponential time complexity is practicable on the parameters we will
propose. However, we decided not to consider this as a threat. Indeed,
if many attacks in the literature on code--based cryptography rest on
the calculation of minimum weight codewords (see for instance
Sidelnikov Shestakov attack \cite{SS92}) here the situation is different
since the distinguisher is based on the absence of low weight codewords.
Such a lack of low weight codewords, despite being practically observable
for the parameters we will propose in Section~\ref{sec:param} does not
look to be exploitable to mount an attack.

\subsection{An attack on the system}
Here we present a key-recovery attack and first give a quick overview.
The goal is, starting from the dual public code $\Dmat^\perp$ to
recover the $\Fqm$--linear structure of the underlying Gabidulin code.
Before giving an overview of the attack, let us discuss a bit further
about the hidden $\Fqm$--linearity.

\subsubsection{About the hidden $\Fqm$--linear structure}
If we had access only to a matrix
description of a Gabidulin code $\Cmat$, the $\Fqm$--linear structure
could be recovered by computing the left--stabilizer algebra of the code,
namely,
\[
  \text{Stab}_L(\Cmat) \eqdef \{\Av \in \Fq^{m \times m} ~|~ \Av \Cmat
  \subseteq \Cmat\}~.
\]
The computation of this algebra boils down to the resolution of a
linear system and the obtained algebra is nothing but a representation
of $\Fqm$. With this algebra at hand, one recovers the $\Fqm$--linear
structure of the code (see \cite{CDG20} for further details). Once this
$\Fqm$--linear vector code structure is recovered, one can apply
classical key recovery attacks on Gabidulin codes
\cite{O06,GF22,CZ23}.

However, in our setting, we only have access to a public code $\Dmat$
which contains a subcode of a matrix Gabidulin code. Thus, it is no longer
$\Fqm$--linear and hence the left stabilizer algebra has no reason to
be isomorphic to $\Fqm$ and will be composed only by scalar matrices
that is to say, with a high probability
\[
  \text{Stab}_L (\Dmat) = \{\alpha \mathbf I_{m} ~|~ \alpha \in \Fq \} \ .
\]
Still, the public code contains a subcode of a Gabidulin code and the
latter has an $\Fqm$--linear structure and hence a non trivial left stabilizer algebra.
The attack to follow targets to recover this algebra which is isomorphic to $\Fqm$.
\medskip 

{\bf \noindent The attack in a nutshell.} The idea of the attack is
that if we can obtain a pair $\Cv, \Cv' \in \Dmat^\perp$ that are both
in the dual Gabidulin code and correspond to $\Fqm$--collinear
codewords, then one can deduce the hidden $\Fqm$--linear structure
from the set of matrices $\{\Pv ~:~ \Pv \Cv = \Cv'\}$ and
$\{\Qv ~:~ \Qv \Cv' = \Cv\}$. The bottelneck is precisely to obtain
such a pair $\Cv, \Cv'$. For this sake, one could draw pairs at random
and repeat the process until it succeeds. We can actually be a bit
more efficient by taking pairs of codewords that share a same row
space (the latter condition being necessary for the corresponding
vectors to be $\Fqm$--collinear).

In the sequel we describe the attack in more details.
\medskip 

{\bf \noindent Context.}
As in Subsection~\ref{ss:distinguisher}, we introduce $\lambda, \mu$ such that
\begin{equation}\label{eq:l_a}
  \ell_a - \ell_s = \lambda m + \mu \quad \text{with}\quad 0 \leq \mu < m~.
\end{equation}
The code $\Dmat^\perp$ satisfies
\iftoggle{llncs}{\(
  \Dmat^\perp = \E \oplus \Rs~,
\)}{
\[
  \Dmat^\perp = \E \oplus \Rs~,
\]
}
where\medskip 
\begin{itemize}\setlength{\itemsep}{5pt}
\item $\Rs$ is a random matrix code of dimension $\ell_s$;
\item $\E$ is random subcode of codimension $\ell_a$ of the Gabidulin
  code $\C^\perp$, which has dimension $rm = (n-\kappa)m$;
\item $\Rs \cap  \C^\perp = \{\vec 0\}$.
\end{itemize}
\medskip 
The core of the attack consists of finding a pair of codewords of
$\mathcal E$ (and hence of $\mathcal{C}^\perp$) that are~$\Fqm$--collinear when regarded as vectors. Finding such a pair is
hard. A first hint for the search is the following statement.

\begin{lemma}\label{lem:row_space}
Let $\Cv, \Cv' \in \mathcal{C}^\perp$ which are $\Fqm$--collinear when
  regarded as vectors. Then the matrices~$\Cv, \Cv'$ have the same row space.
\end{lemma}

\begin{proof}
  If $\Cv, \Cv'$ are $\Fqm$--collinear when regarded as vectors, then
  there exists $\Pv$ in the left stabilizer algebra of
  $\mathcal{C}^\perp$ such that $\Cv' = \Pv \Cv$. Since the left
  stabilizer algebra of $\mathcal{C}^\perp$ is isomorphic (as a ring)
  to $\Fqm$, all its nonzero elements are invertible. Thus, $\Pv$ is
  nonsingular and the previous equality entails that $\Cv, \Cv'$ have
  the same row space. \iftoggle{llncs}{\qed}{}
\end{proof}

{\bf \noindent Presentation of the attack.} We fix a target rank $s$.

\begin{itemize}\setlength{\itemsep}{5pt}
\item[\textbf{Step 1.}] We search a codeword $\Cv \in \Dmat^\perp$ of
  rank $s$, using Algorithm~\ref{algo:low_weight}.
\item[\textbf{Step 2.}] If there exists another codeword
  $\Cv'\in \Dmat^\perp$ non $\Fq$--collinear to $\Cv$ and with the
  same row space as $\Cv$ we store it, otherwise we restart the
  process from the beginning.
\item[\textbf{Step 3.}] Given $\Cv, \Cv'$ we run a routine described
  further.  If the two matrices $\Cv, \Cv'$ are both in the subcode
  $\mathcal E$ of the Gabidulin code and that they are
  $\Fqm$--collinear (when regarded as vectors), then the
  aforementioned routine succeeds and the left stabiliser algebra of
  the hidden Gabidulin code is recovered. Otherwise, restart the
  process from the beginning.
\item[\textbf{Step 4.}] When the hidden $\Fqm$--linear structure is
  recovered, there are various manners to recover the structure of the
  secret key and then to forge signatures.
\end{itemize}
\medskip 

\iftoggle{llncs}{

    \begin{restatable}{proposition}{propositioncomplexity}\label{prop:complexity_structural_attack}
      The complexity of the attack counted in operations in $\Fq$ is
      \begin{align*}
      &\Omega \left(\frac{1}{b} (k^{\omega - 1}mn + k m^\omega +
        m^{\omega+1}n^{\omega - 1})q^{\ell_a + \ell_s - m + b}\right) \quad \text{if}\ s = n-a \\
                &\Omega \left(\left( k^{\omega - 1}m^2 +km^\omega + p\cdot m^{\omega+1}n^{\omega - 1}\right) \cdot P \right) \qquad\quad \text{if}\ n-r+1 \leq s<n-a,
      \end{align*}
    where $\omega$ denotes the complexity exponent of linear algebra,
    \iftoggle{llncs}{
      $p  = q^{m(a+s-n)+b-2}$ and $P  = q^{\ell_s+\ell_a+2-m + (m-s)(n-a-s)-\max(0, b-1-s)}$.
      }{    \begin{itemize}
  \item[\textbullet] $P  = q^{\ell_s+\ell_a+2-m + (m-s)(n-a-s)-\max(0, b-1-s)}$~;
  \item[\textbullet] $p  = q^{m(a+s-n)+b-2}$~.
  \end{itemize}
}
    \end{restatable}
  }{}

  \begin{remark}
    In practice, in all our estimates, the best value for $s$ was\iftoggle{llncs}{}{ always} $n-a-1$.
  \end{remark}


%% file: src/security_moved.tex
\iftoggle{llncs}{
    \section{Further details on the key recovery attack in Section~\ref{sec:secu}\label{sec:app_secu}}
  }{}

\subsection{The target rank $s$}\label{ss:target_s}
According to the previous discussion (see (\ref{eq:l_a}) and below),
the code~$\Dmat^\perp$ has dimension
\[
  k = (r-\lambda -1)m + (m-\mu) \quad \text{with}\quad 0 \leq \mu < m \ .
\]
Set
\begin{equation}\label{eq:ab}
  a \eqdef (r-\lambda -1), \quad b \eqdef m - \mu,
\end{equation}
which
gives
\[
  k = am+b,\quad \text{with} \quad 0 < b \leq m \ .
\]
The
target rank $s$ will be chosen in the range
\begin{equation}\label{eq:range_s}
  s \in [n-r+1, n-a]\ .
\end{equation}
Indeed,
the underlying Gabidulin code $\mathcal{C}^\perp$ has minimum distance
$n-r+1$ and $\mathcal{C}^\perp \subseteq \mathcal E$, codewords of
weight less than $n-r+1$ are rather unlikely in $\Dmat^\perp$. On the other
hand, the discussion in Section~\ref{ss:basic_routine} shows that, by
Gaussian elimination, computing codewords of weight $n-a$ can be done
in polynomial time: for $s = n-a$, Algorithm~\ref{algo:low_weight}
succeeds in a single loop. The trade-off we have to find comes from the following discussion:
\medskip 
\begin{itemize}\setlength{\itemsep}{5pt}
\item When $s$ is close to $n-r+1$, codewords of weight $s$ are hard to find but a pair
  of such codewords are likely to be $\Fqm$--collinear when regarded as vectors;
\item When $s$ is close to $n-a$, codewords of weight $s$ are easy to
  find but a pair of them is unlikely to be $\Fqm$--collinear.
\end{itemize}
\medskip 
We will give a complexity formula for any value of $s \in [n-r+1, n-a]$. In practice,
the best complexities turned out to occur either for $s = n-a$ or $s = n-a-1$.

\subsection{Description of Step 3}\label{ss:compute_Fqm}
Suppose we are given a pair $\Cv, \Cv' \in \Dmat^\perp$ of rank $s$ and with the same row space,
we aim to decide whether they come from $\Fqm$--collinear vectors and if they do,
to deduce the left stabilizer of the hidden Gabidulin code $\mathcal{C}^\perp$:
\begin{equation}\label{eq:algebraA}
  \A \eqdef \{\Pv \in \Fq^{m\times m} ~|~ \Pv \C^\perp = \C^\perp\}~,
\end{equation}
which is a sub-algebra of $\Fq^{m\times m}$ isomorphic to $\Fqm$.

The key idea is simple: if $\Cv, \Cv'$ are $\Fqm$--collinear, then there exist
$\Pv \in \mathbf{GL}_m(\Fq)$ such that
\[
  \Cv' = \Pv \Cv.
\]
Therefore, we compute all the solutions~$\Xv$
of the affine system
\[
  \Cv' = \Xv \Cv~.
\]
Clearly $\Pv$ is a solution of this system. Unfortunately it is not
unique. Indeed, the space of solutions
is nothing but the space of matrices $\Xv$ such that
\(
  (\Xv - \Pv)\Cv = 0.
\)
Equivalently, it is the affine space of matrices of the shape:
\begin{equation}\label{eq:sol_1}
  \{\Pv + \Mv ~|~ \text{RowSpace}(\Mv) \in \ker_L(\Cv)\}~,
\end{equation}
where $\ker_L(\Cv)$ denotes the left kernel of $\Cv$.
Similarly, one can search $\Pv^{-1}$ as the solutions $\Xv$ of~\(
  \Cv = \Xv \Cv'
\)
and get the affine space
\begin{equation}\label{eq:sol_2}
  \{\Pv^{-1} + \Nv ~|~ \text{RowSpace}(\Nv) \in \ker_L(\Cv')\} \ .
\end{equation}





After possibly performing a change of basis (which will consist
of left multiplying the public code by some nonsingular matrix), one can
assume that
\medskip 
\begin{itemize}\setlength{\itemsep}{5pt}
\item $\ker_L (\Cv')$ is spanned by the $n-s$ last elements of the canonical basis.
\end{itemize}
\medskip 
Thus, let us pick nonsingular matrices $\Av, \Bv$ in the solutions spaces
(\ref{eq:sol_1}) and (\ref{eq:sol_2}) respectively. With the
aforementioned change of basis we have
\begin{align*}
  \Pv &= \Av + \Vv \quad \text{for some } \Vv \in
  \Fq^{m\times m}~;\\
  \Pv^{-1} &= \Bv + (\mathbf{0} ~|~ \Wv) \quad \text{for some } \Wv \in
  \Fq^{m\times (m-s)}~.
\end{align*}
Recall at this step that $\Av, \Bv$ have been taken in the solutions
spaces (\ref{eq:sol_1}) and (\ref{eq:sol_2}) respectively. Hence they are
known, while $\Pv, \Vv, \Wv$ are still unknown.
Now note that
\[
   \Pv \Pv^{-1}=\mathbf{I}_m = \Av \Bv 
  + \Vv\Bv + (\vec 0 ~|~ \Av \Wv) +  (\vec 0 ~|~ \Vv\Wv)~.
\]
The sum of the last two terms, namely
$(\vec 0 ~|~ \Av \Wv) + \Vv (\vec 0 ~|~ \Vv\Wv)$ has the same row space
as $(\vec 0 ~|~ \Wv)$ and hence has its $s$ leftmost columns equal to
zero. Consequently, we know that the $s$ leftmost columns of $\Vv\Bv$
equal those of $\mathbf{I}_m - \Av \Bv$.  Regarding $\Vv$ as an
unknown, and since~$\Vv$'s row space is contained in $\ker_L \Cv$
which has dimension $n-s$, this yields a linear system with $m(n-s)$
unknowns in $\Fq$ (the entries of $\Vv$) and $ms$ equations (one per
entry in the $s$ leftmost columns). From (\ref{eq:range_s}),
$s \geq n-r+1$ and, in our instantiations, $r = m - \kappa$ is small
(see Section~\ref{sec:param}), we have $s \geq n-s$ and hence the
latter system is over-constrained and its resolution is very likely to
provide~$\Vv$, which yields $\Pv$.

Once $\Pv$ is obtained, it suffices to compute the matrix algebra
$\Fq[\Pv]$ spanned by $\Pv$. This algebra is very likely to be
the algebra $\A$ of (\ref{eq:algebraA}). In case we only obtain a subalgebra of $\A$, which would
correspond to a subfield of $\Fqm$ we re-iterate the process and
compose the obtained algebras until we get an algebra of dimension
$m$, which will occur after a constant number of iterations.

In \emph{summary} the previous process permits to deduce the left
stabilizer $\A$ of $\C^\perp$ and hence to get the hidden $\Fqm$--linear
structure. Moreover, if $\Cv, \Cv'$ were not $\Fqm$--collinear, then the above
process will probablt fail. Thus it can also be used to decide whether $\Cv, \Cv'$
are $\Fqm$--collinear.

\begin{lemma}\label{lem:complexity_step3}
  The previous procedure permits to decide whether a given pair
  $\Cv, \Cv'$ of rank $s$ and with the same row space come from
  $\Fqm$--collinear vectors.  If they do, the procedure 
  computes the left stabilizer algebra of the hidden Gabidulin code.
  The complexity of the process is
  \[
    O(m^{\omega +1}n^{\omega - 1}) \quad \text{operations in }\Fqm.
  \]
\end{lemma}

\begin{proof}
  Only the complexity should be checked. The bottleneck of the process
  is the resolution of the system $\Cv' = \Xv \Cv$, which has $m^2$ unknowns
  and $mn$ equations. Hence a complexity in~$O(m^2(mn)^{\omega - 1})$. \iftoggle{llncs}{\qed}{}
\end{proof}


\subsection{Step 4. Finishing the attack}\label{ss:finishing}
Once the $\Fqm$--linear structure is computed, one can complete the
key recovery in many different manners. We briefly explain one
possible approach, but claim that many other exist.  Note that
the cost of the routine described below remains negligible compared to
that of obtaining $\Cv, \Cv'$ (as seen in
Section~\ref{ss:compute_Fqm}).

Recall that $\Dmat^\perp = \E \oplus \Rs$, for some random $\Rs$ of
dimension $\ell_s$ and $\E$ being a codimension~$\ell_a$ sub-code of
the matrix Gabidulin code $\C^\perp$.  Suppose we also have access to
the left stabilizer algebra of $\C^\perp$.
Compute:
\[
  \mathcal{A} \Dmat^\perp = \mathcal{A} \mathcal{E} + \mathcal{A} \mathcal{R}_s.
\]
This space can be computed as the span of all the possible products of
elements of a given $\Fq$--basis of $\mathcal{A}$ by an element of a
given $\Fq$--basis of $\Dmat^\perp$. The space
$\mathcal{A} \mathcal{E}$ is the ``$\Fqm$--linear span of
$\mathcal E$'', which is very likely to be $\C^\perp$.  The other term
$\mathcal{A} \mathcal{R}_s$ is a random $\Fqm$--linear code of
dimension~$\leq m\ell_s$. Next, the dual of this code is an
$\Fqm$--linear sub-code of $\C$ of $\Fqm$--codimension at most
$\ell_s$.  Then, the corresponding Gabidulin code $\C$ can be
recovered from the sub-code using Overbeck--like techniques \cite{O05,O05a}.

\subsection{Complexity of the attack}\label{ss:attack_complexity}

\iftoggle{llncs}{
    \propositioncomplexity*}{
\begin{proposition}\label{prop:complexity_structural_attack}
  The complexity of the attack is the following.
  \begin{enumerate}
  \item\label{item:prop_1} When choosing $s = n-a$:
    \[
      \Omega \left(\frac{1}{b} (k^{\omega - 1}mn + k m^\omega +
        m^{\omega+1}n^{\omega - 1})q^{\ell_a + \ell_s - m - b}\right)
      \quad \text{operations in }\Fq \ .
    \]
    \item\label{item:prop_2} When choosing $n-r+1 \leq s<n-a$:
      \[
        \Omega \left(\left( k^{\omega - 1}m^2 +km^\omega + p\cdot m^{\omega+1}n^{\omega - 1}\right) \cdot P \right)
      \quad \text{operations in }\Fq,
    \]
    where
    \begin{align*}
  P & = q^{\ell_s+\ell_a+2-m + (m-s)(n-a-s)-\max(0, b-1-s)}; \\
  p & = q^{m(a+s-n)+b-2} \ .
\end{align*}
  \end{enumerate}
\end{proposition}
}
  \begin{proof}[Sketch of Proof of the case $s=n-a$]
    The calculation of a pair $\Cv, \Cv'$ of rank $s$ with the same
    row space can be done in a single loop of
    Algorithm~\ref{algo:low_weight}. Its cost is
    $O(k^{\omega - 1}mn + km^{\omega})$, according to
    Proposition~\ref{prop:complexity_finding_low_weight}.

    Next, the probability that $\Cv$ is in $\mathcal E$ is
    $q^{- \ell_s}$ since $\mathcal E$ has codimension $\ell_s$ in
    $\Dmat$.

    Let $W$ be the subspace of $\Dmat^\perp$ of matrices whose row
    space in contained in that of $\Cv$ has~$\Fq$--dimension
    $b$. According to Figure~\ref{fig:matrix_with_man_zeroes} the
    computation of $W$ can be regarded as the computation of a
    subspace of matrices whose $a$ leftmost columns are $0$. Therefore
    $W$ is very likely to have codimension $ma$ in $\Dmat^\perp$ and
    hence we very likely have
    \[
      \dim W = b\ .
    \]
    
    The matrix $\Cv'$ was picked uniformly at random in
    $W \setminus \langle \Cv \rangle_{\Fq}$.
    Let us estimate the
    probability that $\Cv'$ is $\Fqm$--collinear to $\Cv$ in
    $W$.
    An element $\Cv'$ which is $\Fqm$--collinear to $\Cv$ without
    being~$\Fq$--collinear to it lies in the dual Gabidulin code $\mathcal{C}^\perp$ and the
    probability that $\Cv'$ lies in
    $\Dmat^\perp \setminus \langle \Cv \rangle_{\Fq}$ is
    $q^{- \ell_a}$.  Consider the space $V \subseteq \Fq^{m \times n}$
    of matrices $\Fqm$--collinear to $\Cv$ and $V_0$ a complement
    subspace of $\langle \Cv \rangle_{\Fq}$ in $V$.  Then $V_0$ has $\Fq$--dimension
    $m-1$ and
    \[
      \mathbb{E} \left(\sharp  V_0 \cap \Dmat^\perp \right) = q^{m-1-\ell_a} \ . 
    \]
    From Markov inequality, for any $1 \leq \delta \leq b-1$
    \[
      \prob \left( \dim V \cap \Dmat^\perp \geq {\delta} +1 \right) =
      \prob \left( \dim V_0 \cap \Dmat^\perp \geq {\delta} \right)
      \leq q^{m-1-\ell_a-\delta}.
    \]
    Denote by $\mathscr{E}$ the event ``$\Cv'$ is $\Fqm$--collinear to $\Cv$''.
    Since $\Cv'$ is drawn uniformly at random in $W \setminus \langle \Cv \rangle_{\Fq}$,
    we have
    \begin{align*}
      \prob \big(\mathscr{E} \big)  
                                   & = \sum_{\delta=1}^{b-1}  \prob \big(\mathscr E  ~|~ \dim V_0 \cap \Dmat^\perp=\delta\big) \prob(\dim V_0 \cap \Dmat^\perp=\delta)\\
             & \leq \sum_{\delta = 1}^{b-1} q^{\delta+1-b} q^{m-1-\ell_a - \delta} =  (b-1) q^{m-\ell_a+b}
    \end{align*}
    In summary, the probability to find $\Cv, \Cv'$ such that $\Cv \in \mathcal \Cv^\perp$ and $\Cv'$ is $\Fqm$--collinear to $\Cv$
    is upper bounded by
    \[
      (b-1)q^{m - \ell_a - b - \ell_s} \ .
    \]
    Hence the average number of times we iterate the process is bounded from below by the inverse of the above probability.
    Moreover, one iteration consists of one loop of Algorithm~\ref{algo:low_weight} followed by the routine described in Section~\ref{ss:compute_Fqm}
    (Step 3).
    Thus, according to Proposition~\ref{prop:complexity_finding_low_weight} and Lemma~\ref{lem:complexity_step3}, the
    complexity of one iteration is $O(k^{\omega - 1}mn + k m^\omega +
        m^{\omega+1}n^{\omega - 1})$.
    \iftoggle{llncs}{\qed}{}
  \end{proof}

  \begin{proof}[Sketch of Proof of the case $s< n-a$]
    We run Algorithm~\ref{algo:low_weight} to get $\Cv$ of rank $s$. From Proposition~\ref{prop:complexity_finding_low_weight}
    the probability to get such
    a $\Cv$ is \[\approx q^{s(m+n-a-s)+ m(a-n)+\max(0, b-1-s)}~.\] Moreover, the probability that it also lies in $\C^\perp$ is $q^{-\ell_s}$.
    Thus, the probability to get a valid $\Cv$ in one loop of Algorithm~\ref{algo:low_weight} is
    \begin{equation}\label{eq:p1}
      \approx q^{s(m+n-a-s)+ m(a-n)+\max(0, b-1-s)-\ell_s}.
    \end{equation}

    Now, here the most likely situation is that, when $\Cv$ is found, the subspace $W$ of $\Dmat^\perp$ of matrices
    whose row space is contained in that of $\Cv$ will have dimension $1$ and be spanned by~$\Cv$. Indeed,~$\Cv$
    is a solution of an over-constrained system. Now we are getting into the unlikely situation where~$W$ has dimension $>1$.

    The overall space $W' \subseteq \Fq^{m\times n}$ of matrices whose
    row space is contained in that of $\Cv$ has dimension $ms$ and
    \[W = \Dmat^\perp \cap W'~.\]
    Choose $W'_0$ a complement
    subspace of $\langle \Cv \rangle_{\Fq}$ in~$W'$. 
    Regarding $\Dmat$ as a random space, the probability that an
    element of~$W'_0$ lies in $\Dmat^\perp$ is
    $q^{\dim \Dmat^\perp - mn} = q^{m(a-n)+b}$. Therefore,
    \begin{align*}
      \mathbb{E} (\sharp W'_0 \cap \Dmat^\perp) &= \sum_{\Dv \in W'_0 }\prob (\Dv \in \Dmat^\perp)
                              =  q^{ms-1}q^{m(a-n)+b} =  q^{m(a+s-n)+b-1}~.
    \end{align*}
    By Markov inequality,
    \begin{equation}\label{eq:p2}
      \prob (\dim  W'_0 \cap \Dmat^\perp \geq 1)  = \prob (\sharp  W'_0 \cap \Dmat^\perp \geq q) 
       \leq q^{m(a+s-n)+b-2}.
    \end{equation}

    This yields the probability that, when a $\Cv$ of rank $s$ is
    found in $\Dmat$, another $\Cv'$ non $\Fq$--collinear to $\Cv$ and
    with the same row space also lies in $\Dmat^\perp$. Still, such a
    $\Cv'$ might not be~$\Fqm$--collinear to $\Cv$.
    
    Actually, using a similar Markov-based reasoning as in the proof of the case
    $s = n-a$, we prove that the probability that a matrix $\Cv'$ that
    is $\Fqm$--collinear to $\Cv$ lies in $\Dmat$ is bounded from
    above by
    \begin{equation}\label{eq:p3}
      q^{m-2-\ell_a}.
    \end{equation}

    In summary, putting \eqref{eq:p1} and \eqref{eq:p3} we will be
    able to find a valid pair $\Cv, \Cv'$ in an average number of
    \begin{align*}
      \approx P & \eqdef q^{\ell_s - s(m+n-a-s)-m(a-n)-\max(0, b-1-s) - m + \ell_a+2}\\
                & = q^{\ell_s+\ell_a+2-m + (m-s)(n-a-s)-\max(0, b-1-s)} \quad \text{iterations of Algorithm~\ref{algo:low_weight}.}
   \end{align*}
   For each iteration, the most likely situation is that either there
   is no $\Cv$ or there is a $\Cv$ and no candidate for $\Cv'$. More
   precisely, when a $\Cv$ is found, the most likely situation is that the subspace of $\Dmat^\perp$ of
   matrices whose row space is contained in that of $\Cv$ has dimension 1 and is spanned by~$\Cv$ itself.
   In such a situation there is no need to run the routine described in Section~\ref{ss:compute_Fqm} (Step 3).
   Still, according to \eqref{eq:p2}, there is a proportion
   \[
     p \eqdef q^{m(a+s-n)+b-2}
   \]
   of cases where a candidate for $\Cv'$ exists without being necessarily $\Fqm$--collinear to $\Cv$. In such
   a case, we have to run the routine described in Section~\ref{ss:compute_Fqm} (Step 3.)
   This yields an additional cost of $O(m^{\omega +1}n^{\omega -1})$
   with a proportion $p$ in the complexity.
   \iftoggle{llncs}{\qed}{}
  \end{proof}


%% file: src/parameters.tex
\section{Parameters}\label{sec:param}

\subsection{How to choose parameters for our scheme}

We discuss here the parameter selection of $\miranda$ when instantiated with Gabidulin codes (with respect to the attacks presented in the previous section). We consider Gabidulin codes whose length is equal to the dimension of the extension field: we choose $n=m$ for all our parameter sets. On the other hand, we consider only binary fields. \iftoggle{llncs}{}{The $128$-, $192$-, and $256$-bits security levels are considered.} To ensure a fair comparison with other NIST-compliant signature schemes, we keep 15-bit security above the required level, that corresponds to levels I, III and V of security.
\medskip

{\bf \noindent Structural constraints on parameters.} Let $t =\frac{n-\kappa}{2}$ the decoding capacity.  $\miranda$ being instantiated as an FDH-signature scheme, there must exist for any syndrome $\vec{s}$ an associated $\vec{E}$ of rank weight $\leq t$ for the public code. This requirement in particular enforces that $t$ must be greater than the \emph{Gilbert–Varshamov radius}, that is the minimum radius for which one can expect at least one solution to the above equation for any~$\vec{s}$. Therefore, we have to choose $\ell_{a}$ large enough to densify the code. But doing so will imply that the (expected) number of possible errors of this rank for a given syndrome is large. In particular, it is approximately equal to: $q^{\ell_a-\ell_s-\left(\frac{n-\kappa}{2}\right)^{2}}. $
\medskip 

{\bf\noindent Attacks on the scheme.} We recall here what are the two attacks we must consider to assert the security of the proposed parameters.
\medskip 

{\it \noindent Forgery attack.} The first attack relies on directly solving a $\minrank$ instance of parameters $(m,n,\kappa m+\ell_{s}-\ell_a,q)$ (as per Definition~\ref{def:minRank}) without exploiting the particular structure of the public key, the\iftoggle{llncs}{}{ structure of the} $\Bv_i$’s being considered \iftoggle{llncs}{random}{ as random}.

Recall it just needs to find one solution of the $\minrank$ instance to get a valid signature. It turns out that in this case, the so-called {\it combinatorial class of attacks} can exploit this case. However, due to the large dimension of the code, the limiting attack is the so called \emph{big-k} algorithm~\cite{GC00}, that already takes into account, in its complexity, the large number of solutions.

There also exists another class of attacks: \emph{algebraic attacks} \cite{BBCGPSTV20,BBBGT22}. However the latter do not exploit the benefit of the exponential number of solutions and it turns out that they are not competitive with big-$k$ in our parameter set.
\medskip 

{\noindent \it Structural attack.} As explained in Section~\ref{sec:secu}, it is also possible to attack $\miranda$ instantiation via its algebraic structure arising from Gabidulin codes which are $\mathbb{F}_{q^{m}}$-linear. The cost of this attack is given Proposition~\ref{prop:complexity_structural_attack}.

\subsection{Practical choice of parameters}

{\bf \noindent Public key size.} Once the $\miranda$ parameters have been set to obtain $\lambda$ security bits for both attacks, here is how the following quantities are computed the public key and signatures sizes.
The public key consists of a matrix code of size $m\times n$ and dimension $k=m(n-\kappa) + \ell_{s}-\ell_{a}$ binary matrices (recall that we set $m=n$). 
It is possible to represent this code as a matrix of size $m^2\times k$ in systematic form, which requires $$\text{Size} (\mathsf{pk}) = \left\lceil\frac{(m\kappa-\ell_s+\ell_a)(m(n-\kappa)+\ell_s-\ell_a)}{8}\right\rceil.$$

{\bf \noindent Signature size.} The signature consists of the error matrix $\vec{E}$ and the random $\texttt{salt}$ with~$\lambda$ bits. But, instead of signing by providing the full matrix $\vec{E}$, it suffices to provide its support (\emph{i.e.,} the $\Fq$--vector space spanned by its columns) which can be represented by $t$ vectors $(\vec{b}_j)_{j\in\lbrace 1, \dots,t\rbrace}\in(\Fq^m)^t$.

This trade-off allows to divide by 2 the size of the signature, at the cost of a slower verification time: rather than simply checking the result of a  matrix-vector multiplication, we must verify that the system $\vec{s} = (\text{tr}(\vec{E} \vec{B}_i^\top))_{i=1}^{m(n-\kappa)-\ell_a + \ell_s},$\iftoggle{llncs}{}{ indeed } has a solution, where each column $\vec{e}^{(i)}$ of the matrix error $\vec{E}$ is itself an $\Fq$-linear combination:
\iftoggle{llncs}{$\vec{e}^{(i)} = \sum_{j=1}^t a_j^{(i)}\vec{b}_j$
}{
	$$\vec{e}^{(i)} = \sum_{j=1}^t a_j^{(i)}\vec{b}_j$$
} whose unknowns are the coefficients $a_j^{(i)}\in\Fq$.

This support can be written in reduced row echelon form. Assume we write this basis as a matrix in $\Fq^{m\times t}$. Since this matrix has rank $\leq t$, there exists a minor of size $\leq t\times t$ on which one can perform linear combinations of columns to obtain the identity matrix.
We can encode the coordinates of these rows as a string of 9 bits (10 for the parameters for $\lambda=256$ bits of security), where each of the 512 possible strings is associated to a set of $t$ rows determined in advance. For the parameters we choose, the probability that none of the $2^9$ combinations suit is inferior to $2^\lambda$.
Therefore, for a salt of size $\lambda$, the signature size in bytes is: 
$$\text{Size} (\sigma) = \left\lceil\frac{\frac{n-\kappa}{2}\left(m-\frac{n-\kappa}{2}\right)+9+\lambda}{8}\right\rceil~\cdot$$

{\bf\noindent A trade-off allowing to accelerate the signature verification.} We said that an other trade-off is possible by including the error in its entirety in the signature.
In this case, we should send the full vector of coordinates of the error in this given basis. We can obtain a signature of size $ \left\lceil\frac{\frac{n-\kappa}{2}\left(2m-\frac{n-\kappa}{2}\right)+9+\lambda}{8}\right\rceil$.
This would have doubled the signature size, but is greatly accelerated its verification: it then consists only a matrix-vector multiplication.
\medskip 

{\bf\noindent Resulting parameters.} We propose several sets of parameters. In tables below, the abbreviation ``Dens'' refers to the $\log_2$ of the average density of decodable vectors in the Gabidulin code.
We considered the complexities given by the estimators provided by the TII \cite{tii} for the $\minrank$ instance.
We propose in Table \ref{mir128} our parameter sets for the short version of the signature\iftoggle{llncs}{.}{, compliant with levels I, III, V of the NIST requirements.} 

\begin{table}[h]
\begin{center}
{\setlength{\tabcolsep}{0.3em}
    {\renewcommand{\arraystretch}{1.6}
        {\scriptsize
            \begin{tabular}{|c|c|c|c|c||c|c|c||c|c|c|}
                \hline
                $m$ & $\kappa$ & $\; t\;$ & $\ell_a$ & $\;\ell_s\;$ & Dens. & Forge. & Struc. & $\sigma$ (B) & $\mathsf{pk}$ (B) & $\lambda$\\ \hline\hline
                79 & 71 & 4 & 189 & 2 & -18 & 146 & 146 & 55 & 346 K & 128 \\ \hline
                113 & 107 & 3 & 220 & 3 & -12 & 145 & 143 & 59 & 737 K & 128 \\ \hline
                127 & 121 & 3 & 250 & 2 & -11 & 156 & 160 & 64 & 1036 K & 128 \\ \hline
                223 & 219 & 2 & 330 & 4 & -8 & 149 & 147 & 73 & 3518 K  & 128 \\\hline\hline
                127 & 119 & 4 & 315 & 2 & -18 & 221 & 217 & 87 & 1417 K & 192\\ \hline
                173 & 167 & 3 & 340 & 4 & -13 & 207 & 207 & 89 & 2626 K & 192 \\ \hline
                179 & 173 & 3 & 350 & 2 & -11 & 210 & 212 & 92 & 2895 K & 192 \\ \hline\hline   
                163 & 155 & 4 & 400 & 2 & -18 & 277 & 276  & 113 &  3009 K & 256\\ \hline
                233 & 227 & 3 & 160 & 5 & -14 & 270 & 2704 & 120 & 6393 K & 256 \\ \hline
                239 & 233 & 3 & 470 & 4 & -12 & 277 & 273 & 122 & 6912 K & 256 \\ \hline
            \end{tabular}
            \vspace{0.2\baselineskip}
}}}
\end{center}
	\caption{Parameters\iftoggle{llncs}{}{ for $\miranda$, short version,} taking  $\omega = 2.8$ as the complexity exponent of linear algebra.}
	\label{mir128}
\end{table}

Note that the parameters have been chosen such that the statistical distance between the distribution of outputs of $\mathsf{InvertAlg}$ and the uniform one over~$\mathcal{B}_t^{m, n, q}$ (as given in Theorem~\ref{theo:unifPreimage}) is less than $2^{-64}$\iftoggle{llncs}{.}{, in order to comply with the requirements of the NIST.}

\subsection{Performances}

The key generation (which refers to the $\mathsf{Trapdoor}$ algorithm) consists of determining the systematic form of a matrix code with parameters~$[m\times m,2tm-\ell_a+\ell_s]$. This can be done in $O(m^2(2tm-\ell_a+\ell_s)^2)$ binary operations.
As explained in Proposition \ref{propo:decoGab}, the decoding of a Gabidulin code of length $m$ requires $O(m^2)$ operations in $\Fqm$, which corresponds to $O(m^3\log m)$ binary operations (since we only consider parameter sets with $q=2$). The average number of syndromes we try to decode is given by the inverse of the density of decodable syndromes. We deduce that the average running time of the signing algorithm is:
$$
O(m^3\log m2^{t^2+\ell_s}) \ .
$$
The signature verification algorithm consists of checking that an over-determined system of $2tm-\ell_a+\ell_s$ equations in $\F$ with $tm$ unknowns does indeed have a solution. This can be done by performing a Gaussian elimination on a binary matrix of size $tm\times (2tm-\ell_a+\ell_s)$, which requires a number of operations equal to: 
$$
O(m^2t^2(2tm-\ell_a+\ell_s)).
$$ 
We did a basic implementation of $\miranda$ on a computer whose processor is CPU Intel Core Ultra 7 265U. We \iftoggle{llncs}{give}{consider} the \iftoggle{llncs}{different running times for the parameter set given by $m=113$ (see Table~\ref{mir128})}{ parameter set given by $m=113$ which leads to a signature of 59 B (see Table~\ref{mir128}) to give an example for execution time.}


This basic implementation leads to a decoding of 0.75ms by using the standard decoding algorithm of Gabidulin code provided by the RBC library, adapted to small values of $m$ since multiplication in $\mathbb{F}_{2^m}$ is hardcoded. Considering the $2^{9}$ average attempts of decoding by guessing a vector $\vec{t}$ such that $(\sv\vert\vec{t})$ is a decodable syndrome, we obtain an average time of approximately $400$ms. There exists syndrome algorithms based on the direct decoding of the syndrome, which allow to gain a factor 4 on the decoding time. Note that parallelizing the algorithm is possible, and make it much faster by doing independent guesses for $\vec{t}$.

\begin{remark}
	The execution time of the signature algorithm depends heavily on the density of decodable syndromes. For example, it requires only $50$ms for the parameter with $m=223$ thanks to his low density of $2^{-4}$.
\end{remark}

Recall we propose two possible versions of the signature. The first is a short version, in which only the support of the error $\Ev$ is given. As explained above, the verification consists in a gaussian elimination of a binary matrix. Such computation, that can be done with classic algorithms, requires in average $0.5$ms for our parameter with $m=113$. The second is to include the error in its entirety in the signature. This double the signature size, but would have greatly accelerated its verification, as it would have consisted of a simple matrix-vector multiplication. Verification time can be reduced to approximately 0.05 ms.

\iftoggle{llncs}{}{We summarize the time execution of these algorithms in Table \ref{exec_time}.}

\begin{table}[h]
	 \begin{center}
	 	 {\setlength{\tabcolsep}{0.3em}
	 	 	 {\renewcommand{\arraystretch}{1.6} 
	 	 	 	\begin{tabular}{|c|c|c|c|c|}
	 	 	 		 \hline Version & $\mathsf{KeyGen}$ & $\mathsf{Sign}$ & $\mathsf{Verif}$ & $\sigma$ (B) \\ \hline 
	 	 	 		 \textbf{Short} & $20$ms & $400$ms & 0.5 ms & 59 \\ \hline
	 	 	 		 \textbf{Fast} & $20$ms & $400$ms & 0.05 ms & 118 \\ \hline
	 	 	 	\end{tabular} 
	 	 \vspace{0.2\baselineskip} }}\\
 	 \end{center}
 	 \caption{Execution time of algorithms related to $\miranda$\iftoggle{llncs}{ when $m = 113$}{, for the parameter set given by: $m=113$}}
 	 \label{exec_time}
 	\end{table}

\subsection{Comparison with other signature schemes}

\iftoggle{llncs}{}{\mbox{ }
	
	\medskip 
	{\bf \noindent Comparison with \textsf{CFS} scheme.}
	We can explain why the parameters of our Miranda signature scheme are better than those of \textsf{CFS} \cite{CFS01}. We fix $t$ the weight of the error to be decoded in the \textsf{CFS} signature. Let $m$ be the parameter associated with the considered Goppa code. As explained in \cite{F10}, the best known attack against the \textsf{CFS} signature is Bleichenbacher's attack based on the birthday paradox, and has a complexity of about $O(2^{mt/3})$. Denoting $O(2^\lambda)$ the value that this must reach according to the security parameter $\lambda$, we deduce that $m$ is linear in $\lambda$. However, the public key corresponding to the description of a code with parameters $[2^m,mt]$ has size exponential in $\lambda$.
	Conversely, the complexity of this attack on our Miranda signature is on the order of $O(2^{kt})$, where the weight $t$ of the error to be decoded is fixed. The size of the Miranda public key is on the order of $O(k^3)$. We deduce that the size of the public key of our Miranda signature is cubic in the security parameter.
	
	{\medskip 
		\bf \noindent Comparison with other schemes.}
}
\iftoggle{llncs}{}{We propose in Table~\ref{size-comp128} a comparison of our $\miranda$ scheme for 143 bits of security with other hash-and sign signature schemes.}

\begin{table}[!h]
	\begin{center}
		{\setlength{\tabcolsep}{0.3em}
			{\renewcommand{\arraystretch}{1.6}
				{\scriptsize
					\begin{tabular}{|c||c|c|}
						\hline
						Scheme & $\sigma$ (B) & $\pk$ (B) \\ \hline\hline
						\textbf{$\miranda$} Tab. \ref{mir128} & 60 & 734 K \\ \hline
						UOV \cite{uov}  & 96 & $412$k \\ \hline
						Mayo \cite{mayo} & 186 & $4.9$k \\ \hline
						Falcon \cite{FHKLPPRSWZ} & 666  & $897$ \\ \hline
						Wave \cite{BCCCDGKLNSST23} & 822 & 3$.6$M \\ \hline
					\end{tabular}
					\vspace{0.5\baselineskip}
		}}}
	\end{center}
\end{table}

%% file: src/encryption_scheme.tex
\vspace*{-0.5cm}
\section{McEliece-like encryption scheme relying on the $\addremove$ masking}

It is also possible to design a McEliece-like encryption scheme scheme based on the $\addremove$ masking using a Gabidulin matrix code: we instantiate McEliece's scheme with 
the family $\addremove (\mathcal{F},\ell_a,\ell_s)$, where $\mathcal{F}$ corresponds to the set of Gabidulin codes $\mathsf{Gab}(\vec{g}, \kappa)$ seen as  matrix codes. Roughly speaking the idea when selecting parameters is to exchange the role of the ``add'' $\ell_a$ and ``remove'' $\ell_{s}$ in $\miranda$ and consider a decoding distance $t$ such that there is unicity of the solution. Furthermore, we will chose $\ell_a$ as $0$ because in the decoding process we have to perform an exhaustive search over the code $\mathcal{A}$ which has size $q^{\ell_a}$ (small values of $\ell_{a}$ are also admissible but it will increase the running time of the decryption). Notice that in $\miranda$ this is $\ell_s$ which is chosen as small. 




The public-key of the scheme therefore consists in selecting a subcode $\C_s$ of  $\C =M_{\mathcal{B}}(\mathsf{Gab}(\vec{g}, \kappa))$ with codimension $\ell_{s}$. More precisely, the key generation consists in sampling a 
$\Qv$ from $GL_{m^2}(\Fq)$, $\Hv_s$ be a random parity check matrix of the unfolded code associated to $\C_s$. 
We compute $\pk=\Hv_s\Qv$.
The required information to decode a syndrome is $\sk = (\mathsf{Gab}(\vec{g}, \kappa),\gamma,\Qv,\mathcal{B})$. Indeed, the decryption of the syndrome $\sv\in\Fq^{(m-\kappa)m+\ell_s}$ consists in decode $\Qv^{-1}\sv$ (where $\sv$ is written as a column) using the decoding algorithm of $M_{\mathcal{B}}(\mathsf{Gab}(\vec{g}, \kappa))$.

Since it is a McEliece encryption in its dual form, the message to encrypt is hashed into a matrix of rank at most $t$, and the ciphertext is its associated syndrome, that is a vector of length $m(n-\kappa)+\ell_s$. The size of the public key is the same as the signature scheme. The structural attack is also the same, except that we seek to recover the $\Fqm$-linearity directly in $\C_s$ rather in its dual as we exchanged the role of $\ell_{a}$ and $\ell_s$. 
The resulting parameters are given in Table \ref{enc_scheme}. Let us stress that we choose our parameters such that the decrypting without the secret key costs $2^{\lambda}$ (according to our attacks) but such that we can distinguish the public-code from a random code with a cost less than $2^{\lambda}$ (even if this distinguisher is exponential by nature). This avoids to reduce the security to a pure a $\minrank$ problem  but it does not break the scheme. 

\begin{table}[h]
	\begin{center}
		{\setlength{\tabcolsep}{0.3em}
			{\renewcommand{\arraystretch}{1.6}
				{\scriptsize
					\begin{tabular}{|c|c|c|c||c|c||c|c|}
						\hline
						$m$ & $\kappa$ & $\; t\;$ & $\;\ell_s\;$ & Forge & Struc. & $\mathsf{ct}$ (B) & $\mathsf{pk}$ (B) \\ \hline\hline
						37 & 25 & 6 & 171 & 141 & 143 & 77 & 58 K \\ \hline
						41 & 31 & 5 & 159 & 158 & 143 & 72 & 79 K \\ \hline
						43 & 35 & 4 & 164 & 146 & 160 & 64 & 85 K \\ \hline
						53 & 47 & 3 & 188 & 153 & 193 & 64 & 146 K \\ \hline\hline
						53 & 43 & 5 & 215 & 206 & 208 & 95 & 194 K \\ \hline
						59 & 51 & 4 & 208 & 214 & 207 & 85 & 238 K \\ \hline
						71 & 65 & 3 & 195 & 211 & 208 & 78 & 343 K \\ \hline\hline
						67 & 57 & 5 & 320 & 288 & 272 & 124 & 433 K \\ \hline
						97 & 91 & 3 & 360 & 288 & 272 & 118 & 997 K \\ \hline
					\end{tabular}
					\vspace{0.2\baselineskip}
		}}}
	\end{center}
	\caption{Parameters for the related encryption scheme, $\ell_a=0$, $\omega=2.8$}
	\label{enc_scheme}
\end{table}

With these parameters the decryption algorithm is extremely fast: by choosing $\ell_{a} = 0$ we have no exhaustive search to perform over $\mathcal{A}$,  contrary to the signature generation in $\miranda$, and the decryption only involves {\it one} decoding.

These parameters lead to an encryption and decryption times of less than $1$ms, similar to those of the encryption scheme of \cite{ACDGV24}.

\begin{remark}
	The masking proposed here is different from that of \cite{ACDGV24}. The (exponential-time) cryptanalysis proposed by \cite{PWL25} which successfully break parameters of \cite{ACDGV24} does not apply in our case.
\end{remark}


%% file: src/acknowledgements.tex
\section*{Acknowledgments.}
The authors express their deep gratitude to Anthony Fraga for providing us practical running times for the decoding of Gabidulin
codes. The third and fourth authors thank Romaric Neveu for fruitful conversation. The authors are supported by the French \emph{Agence Nationale de la Recherche} (ANR) through the \emph{Plan France 2030 programme}
ANR-22-PETQ-0008 ``{PQ-TLS}''. Thomas Debris-Alazard was Funded by the European Union (ERC, IQ-SCALe, 101219089). Views and opinions expressed are however those of the author(s) only and do not necessarily reflect those of the European Union or the European Research Council. Neither the European Union nor the granting authority can be held responsible for them.  Alain Couvreur is
partially funded by the french ANR grant \emph{Projet de recherche
  collaboratif} ANR-21-CE39-0009 ``{BARRACUDA}'' and by
Horizon--Europe MSCA-DN project ``{ENCODE}''.


%% file: src/conclusion.tex
\section{Conclusion}

Despite its large public keys and slightly long signing times, $\miranda$ remains fully practical. It also gathers the following strengths: no rejection sampling; very short signatures with possible applications such as blind signatures; other potential practical applications for blockchain or advanced cryptographic protocols.

We propose here the parameters that should be chosen for \textsf{CFS} to reach a security of 143 bytes, according to the formula for the attack complexity given by \cite{F10}. We will start by fixing $t=10$, which gives the density of decodable syndromes (the authors of \textsf{CFS} choose values ranging from~$8$ to $10$, and we choose $10$ to achieve a smaller public key), and then fix $m$ which allows reaching the required security level. We recall that if the public key is the description of a $[2^m,mt]$ linear code, and the signature consists of a vector of Hamming weight~$t$ of length $2^m$. The coordinates of 1 of this vector can be encoded in $b$ bits, where~$\binom{2^m}{t}<2^b$ (on the other hand, a counter value also needs to be stored, which averages around $t!$). It gives a signature size of approximately~$O(mt)$ bits, that is, contrary to the public key size, linear in the security parameter $\lambda$. Even if the signature sizes remain smaller than that of our Miranda protocol, the unreasonably large size of the public key required to achieve similar levels of security makes the scheme impractical.

%% file: src/appInvertF.tex
\section{Proof of Lemma \ref{lemma:invertF}}\label{app:distStat} 

\begin{proof}
  By definition of the statistical distance and since $\prob(x_u=\bot) = 0$ we get
  \[
    \Delta (x, x_u) = \left| \prob(x = \bot) \right| + \sum_{x_0\in \mathcal D} \left| \prob (x=x_0) - \prob(x_u=x_0)\right|~.
  \]
  Therefore,
  \begin{equation}\label{eq:distance_stat}
    \Delta (x, x_u) =   \left| \prob(y \notin f(\mathcal{D})) \right|+ \sum_{y_0 \in f(\mathcal{D})} \sum_{x_0 \in f^{-1}(\{y_0\})} \left| \prob (x=x_0) - \prob(x_u=x_0)\right|~. 
  \end{equation}
  Now, note that for a given $y_0 \in \mathcal F$, the quantity
  $\left| \prob (x=x_0) - \prob(x_u=x_0)\right|$ is the same for any~$x_0 \in f^{-1}(\{y_0\})$.
  Consequently, fix $y_0 \in \mathcal F$ and $x_1 \in f^{-1}(\{y_0\})$, then
  \begin{align*}
    \sum_{x_0\in f^{-1}(\{y_0\})} \left| \prob (x=x_0) - \prob(x_u=x_0)\right| & =
                                                                                 \sharp f^{-1}(\{y_0\}) \left| \prob (x=x_1) - \prob(x_u=x_1)\right| \\
                                                                               & =  \left| \sharp f^{-1}(\{y_0\}) \prob (x=x_1) - \sharp f^{-1}(\{y_0\}) \prob(x_u=x_1)\right|\\
                                                                               & = \left| \prob(x \in f^{-1}(y_0)) - \prob(x_u \in f^{-1}(y_0)) \right| \\
    & = \left| \prob(y = y_0) - \prob(f(x_u)= y_0) \right|    
  \end{align*}
  Back to (\ref{eq:distance_stat}), we deduce:
  \begin{align*}
    \Delta (x,x_u) & = \left| \prob(y \notin f(\mathcal D)) \right|
                     + \sum_{y_0 \in f(\mathcal D)} \left| \prob(y = y_0) - \prob(f(x_u)= y_0) \right|\\
    & = \sum_{y_0 \in \mathcal F} \left| \prob(y = y_0) - \prob(f(x_u)= y_0) \right| = \Delta(y, f(x_u))~.
  \end{align*}\iftoggle{llncs}{\qed}{}
\end{proof}
 

%% file: src/app_proof_average_time.tex
\section{Proof of Proposition~\ref{propo:runningTime}}\label{sec:runningTime}

\begin{proof}
	First notice that all the 
	$$
	\left( \tr\left( \vec{E}\vec{B}_{i}^{\top} \right) \right)_{i=1}^{m(n-\kappa) + \ell_{s}}
	$$
	for $|\vec{E}| \leq (n-\kappa)/2$ are distinct as the minimum distance of $\mathcal{C}_{s} \subseteq \mathsf{Gab}(\vec{g},\kappa)$ is greater than $n-\kappa+1$. Therefore there are $\sharp \mathcal{B}_{(n-\kappa)/2}$ such vectors of traces. Let $\mathcal{S}$ denote this set of vectors.

	Given an input $\vec{s} \in \Fq^{m(n-\kappa)+\ell_{s}-\ell_{a}}$, let $G_{\vec{s}}$ be the average number of trials in Instruction~\ref{inst:uniformT} before guessing $\vec{t} \in \Fq^{\ell_{a}}$ such that the algorithm terminates, {\it i.e.,} before guessing $\vec{t}$ such that
	$$
	(\vec{s},\vec{t}) \in \mathcal{S}\ .
	$$
	By definition, 
	$$
	1/G_{\vec{s}} = \mathbb{P}_{\vec{t}}\left( (\vec{s},\vec{t}) \in \mathcal{S} \right)
	$$
	and as $\vec{s},\vec{t}$ are uniform 
	\begin{align*}
		\mathbb{E}_{\vec{s}}\left( 1/G_{\vec{s}} \right)& = \mathbb{P}_{\vec{t},\vec{s}}\left( (\vec{s},\vec{t}) \in \mathcal{S}  \right) = \frac{\sharp \mathcal{S}}{q^{m(n-\kappa)+\ell_{s}}} \\ &=  \frac{\sharp \mathcal{B}_{(n-\kappa)/2}}{q^{m(n-\kappa)+\ell_{s}}} = \Theta\left( q^{  \frac{\left( n-\kappa\right) \cdot \left( m + n  + \kappa\right)}{4} - \ell_{s} } \right)  \ .
	\end{align*} 
	To conclude, notice that each iteration of Algorithm~\ref{algo:sgn} costs $O(m^{3}n^{3})$ 
	which corresponds to the cost for computing~$\vec{Y}$ in Instruction~\ref{inst:Y} (it is the dominant cost).\iftoggle{llncs}{\qed}{}
\end{proof}

%% file: src/appProofLHL.tex
\section{Proof of Proposition \ref{proposition:domainSampleability}}\label{app:proofPropoLHL}

	The proof of Proposition~\ref{proposition:domainSampleability} will rely on the leftover hash lemma \cite{BDKPPS11} (for a proof of this statement in its current form see for instance \cite[Ch.~2, \S 2.5, Lem.~2.5.1]{D23}). 
    \begin{lemma}[Leftover hash lemma]\label{lemme:leftOver} Let
      $E, F$ be finite sets. Let $\mathcal{H} = (h_i)_{i \in I}$ be a
      finite family of applications from $E$ in~$F$. Let $\varepsilon$
      be the ``collision bias'' defined as:
	\begin{displaymath}
		\mathbb{P}_{h,e,e'}(h(e)=h(e')) = \frac{1}{\sharp F} (1 + \varepsilon)
      \end{displaymath}
	where $h$ is uniformly drawn in $\mathcal{H}$, $e$ and $e'$ are independent random variables taking their values~$E$ and following some distribution $\mathscr{D}$. Let $u$ be a random variable uniformly distributed over $F$. We have,
	\begin{displaymath}
		\mathbb{E}_{h}\left( \Delta(h(e), u) \right) \leq \frac{1}{2} \; \sqrt{\varepsilon}.
	\end{displaymath}
\end{lemma}

The proof of Proposition \ref{proposition:domainSampleability} will be a simple combination of the above lemma with the following lemmas.

\begin{lemma}\label{lemma:probFund}
	Let $\mathcal{F}$ be a set of $\lbrack m \times n, k\rbrack_{q}$-codes 
	with minimum distance $\geq d$. 
	Let $\left( \vec{B}_{1},\dots,\vec{B}_{mn-k}\right)$ be chosen uniformly at random among all dual bases of codes in $\mathcal{F}$. We have,
	\iftoggle{llncs}{
	\begin{multline*}
			\forall \vec{s} \in \mathbb{F}_{q}^{m(n-k)}, \; \forall \vec{X} \in \mathbb{F}_{q}^{m \times n}	\mbox{ such that } 0 < |\vec{X}|< d, \\ \mathbb{P}_{\left( \vec{B}_{i}\right)_{i}}\left( \tr\left( \vec{X}\vec{B}_{i}^{\top}\right)_{i} = \vec{s} \right) = \left\{ \begin{array}{cl}
			\frac{1}{q^{mn-k}-1} & \mbox{ if $\vec{s} \neq \vec{0}$} \\
			0 & \mbox{ otherwise.} 
		\end{array} \right.
	\end{multline*}
	}{
		\begin{multline*}
	\forall \vec{s} \in \mathbb{F}_{q}^{m(n-k)}, \; \forall \vec{X} \in \mathbb{F}_{q}^{m \times n}	\mbox{ such that } 0 < |\vec{X}|< d, \\ 
	\mathbb{P}_{\left( \vec{B}_{i}\right)_{i=1}^{mn-k}}\left( \left( \tr\left( \vec{X}\vec{B}_{i}^{\top}\right)\right)_{i=1}^{mn-k} = \vec{s} \right) = \left\{ \begin{array}{cl}
		\frac{1}{q^{mn-k}-1} & \mbox{ if $\vec{s} \neq \vec{0}$} \\
		0 & \mbox{ otherwise} 
	\end{array} \right.
		\end{multline*}
} 
\end{lemma}


\begin{proof}
	First, notice that $\left( \vec{B}_{i} \right)_{i=1}^{mn-k}$ distribution is invariant by any non-singular matrix $\vec{T} = (t_{i,j}) \in \mathbb{F}_{q}^{(mn-k)\times (mn-k)}$, {\it i.e., } $\left( \sum_{i} t_{i,j}\vec{B}_{i}\right)_{j}$ is another dual basis of the code defined by $\left( \vec{B}_{i} \right)_{i=1}^{mn-k}$. Therefore\iftoggle{llncs}{ (as $\vec{T}$ is non-singular)}{},  
	\begin{align}
		\mathbb{P}_{\left( \vec{B}_{i} \right)_{i=1}^{mn-k}}\left( \tr\left( \vec{X}\vec{B}_{i}^{\top}\right)_{i} = \vec{s} \right) &= \mathbb{P}_{\left( \vec{B}_{i} \right)_{i=1}^{mn-k}} \left( \tr\left( \sum_{i=1}^{mn-k} t_{i,j}\vec{X}\vec{B}_{i}^{\top} \right)_{j} = \vec{s}\vec{T} \right) \iftoggle{llncs}{\nonumber \\}{\quad \left( \mbox{as $\vec{T}$ is non-singular}\right)\nonumber \\}
		&= \mathbb{P}_{\left( \vec{B}_{i} \right)_{i=1}^{mn-k}}\left( \tr\left( \vec{X}\vec{B}_{i}^{\top}\right)_{i=1}^{mn-k} = \vec{s}\vec{T} \right)~. \label{eq:invS} 
	\end{align}
	Given $\vec{s}_{1}, \vec{s}_{2} \in \mathbb{F}_{q}^{mn-k}$ which are both non-zero, it exists $\vec{T} \in \mathbf{GL}_{mn-k}(\mathbb{F}_{q})$ such that,
	$$
	\vec{s}_{1}\vec{T} = \vec{s}_{2}\ .
	$$
	Therefore, from \eqref{eq:invS}, 
	$$
	\mathbb{P}_{\left( \vec{B}_{i} \right)_{i=1}^{mn-k}}\left( \tr\left( \vec{X}\vec{B}_{i}^{\top}\right)_{i=1}^{mn-k} = \vec{s}_1 \right) =  	\mathbb{P}_{\left( \vec{B}_{i} \right)_{i=1}^{mn-k}}\left( \tr\left( \vec{X}\vec{B}_{i}^{\top}\right)_{i=1}^{mn-k} = \vec{s}_2 \right).
	$$
	Notice that if $\vec{s} = \vec{0}$ we have that, 
	$$
	\mathbb{P}_{\left( \vec{B}_{i} \right)_{i=1}^{mn-k}}\left( \tr\left( \vec{X}\vec{B}_{i}^{\top}\right)_{i} = \vec{0} \right) = 0
	$$
	as $0 < |\vec{X}| < d$ where $d$ is lower than the minimum distance of the code admitting as dual basis~$\left( \vec{B} \right)_{i=1}^{mn-k}$. It concludes the proof. \iftoggle{llncs}{\qed}{}
\end{proof}

	\begin{lemma}\label{lemma:probaCol}
		Let $m,n,k,q,\ell_a,\ell_s$ be integers and~$\mathcal{F}$ be a set of
		$[m \times n,k]_{q}$-codes with minimum distance $\geq d$ and let $t < d/2$. Let 
		$\mathcal{C} \Unif \mathsf{AddRemove}\left( \mathcal{F},\ell_a,\ell_s \right)$ (as per Definition \ref{def:addRemove}) and~$(\vec{B}_{i})_{i=1}^{mn - k + \ell_{s} - \ell_{a}}$ be a random basis of $\mathcal{C}^{\perp}$. We have,
%
	\begin{multline*} 
	\mathbb{P}_{\left( \vec{B}_{i} \right)_{i=1}^{mn-k+\ell_{s}-\ell_{a}}, \vec{X}, \vec{Y}}\left( \left( \tr\left( \vec{X}\vec{B}_{i}^{\top}\right)\right)_{i=1}^{mn-k+\ell_{s}-\ell_{a}} = \left(\tr\left( \vec{Y}\vec{B}_{i}^{\top} \right)\right)_{i=1}^{mn-k+\ell_{s}-\ell_{a}} \right) \\
	=  \frac{1}{q^{mn-k+\ell_s - \ell_a}}\left( 1 +  O\left( \frac{q^{mn-k+\ell_s - \ell_a}}{\sharp \mathcal{B}_{t}^{m,n,q}} \right) \right) 
	\end{multline*} 
	where $\vec{X},
    \vec{Y}\Unif \mathcal{B}_{t}^{m,n,q}$. 
\end{lemma}

\begin{proof}
	In what follows all probabilities will be computed with the random variables $ \left( \vec{B}_{i} \right)_{i=1}^{mn-k+\ell_{s}-\ell_{a}},$ $\vec{X}$ and $\vec{Y}$. First, by the law of total probabilities,
	\begin{align}
		\mathbb{P}&\left( \left( \tr\left( \vec{X}\vec{B}_{i}^{\top}\right)\right)_{i=1}^{mn-k+\ell_{s}-\ell_{a}} = \left(\tr\left( \vec{Y}\vec{B}_{i}^{\top} \right)\right)_{i=1}^{mn-k+\ell_{s}-\ell_{a}} \right) \nonumber \\
		&= \mathbb{P}\left( \left(\tr\left( (\vec{X}-\vec{Y})\vec{B}_{i}^{\top}\right)\right)_{i=1}^{mn-k+\ell_{s}-\ell_{a}} =\vec{0} \mid \vec{X} \neq \vec{Y} \right)\mathbb{P}\left( \vec{X} \neq \vec{Y} \right) \iftoggle{llncs}{\nonumber\\}{+ \mathbb{P}_{\vec{X}, \vec{Y}}\left( \vec{X} = \vec{Y} \right)\nonumber\\}
		\iftoggle{llncs}{&\qquad\qquad\qquad\qquad\qquad\qquad\qquad\qquad \qquad\qquad\qquad\qquad\qquad + \mathbb{P}_{\vec{X},\vec{Y}}\left( \vec{X} = \vec{Y} \right)\nonumber\\}{}
		&= \mathbb{P}\left( \left( \tr\left( (\vec{X}-\vec{Y})\vec{B}_{i}^{\top}\right)\right)_{i=1}^{mn-k+\ell_{s}-\ell_{a}} =\vec{0} \mid \vec{X} \neq \vec{Y} \right)\left( 1- \frac{1}{\sharp \mathcal{B}_{t}^{m,n,q}} \right) + \frac{1}{\sharp \mathcal{B}_{t}^{m,n,q}}~\cdot \label{eq:probaCol} 
	\end{align}
	By assumption, $\left( \vec{B}_{1},\dots,\vec{B}_{mn - k + \ell_{s}- \ell_{a}} \right)$ is by construction a basis of the dual of 
	$$
	\mathcal{C} = \mathcal{C}_s \oplus \mathcal{A}
	$$
	where $\mathcal{C}_{s}$ is a subcode with codimension $\ell_{s}$ in a code from $\mathcal{F}$. 
	Therefore, 
	$$
	\mathcal{C}^{\perp} \subseteq \mathcal{C}_s^{\perp}.
	$$
	Let us complete the random $\left( \vec{B}_{i} \right)_{i=1}^{mn-k+\ell_{s}-\ell_{a}}$ basis  into a basis of $\mathcal{C}_{s}^{\perp}$ with the help of $\ell_{a}$ matrices $\left( \vec{B}_{mn -k +\ell_{s} -\ell_{a} + 1},\dots,\vec{B}_{mn - k + \ell_s} \right)$. By randomizing we obtain, 
%
	\begin{multline*} 
		\mathbb{P}\left( \left(\tr\left( (\vec{X}-\vec{Y})\vec{B}_{i}^{\top}\right)\right)_{i=1}^{mn-k+\ell_{s}-\ell_{a}} =\vec{0} \mid \vec{X}\neq \vec{Y} \right) 
		\\
		=\sum_{\substack{\vec{s}\in \mathbb{F}_{q}^{mn-k+\ell_{s}}:\\ s_1=\dots= s_{mn-k + \ell_{s} -\ell_{a} } = 0}} \mathbb{P} \left( \left( \tr\left( (\vec{X}-\vec{Y})\vec{B}_{i}^{\top} \right)\right)_{i=1}^{mn-k+\ell_{s}} = \vec{s} \mid \vec{X} \neq \vec{Y}\right)~.
	\end{multline*} 
	In the above sum $\vec{s}$ cannot be equal to $\vec{0}$. Indeed, $0< |\vec{X} - \vec{Y}|\leq 2t < d$ where $d$ is smaller than the minimum distance of $\mathcal{C}_s$. Therefore
	by using Lemma \ref{lemma:probFund},
	\begin{align*} 
		\mathbb{P}
		( \tr\left( 
		(\vec{X}-\vec{Y})\vec{B}_{i}^{\top}\right)_{i} &=\vec{0} \mid \vec{X} \neq \vec{Y} ) = \frac{q^{\ell_{a}}-1}{q^{mn-k+ \ell_{s}}-1}~\cdot
	\end{align*} 
	Plugging this into Equation \eqref{eq:probaCol} leads to:
	\begin{align}
		\nonumber \mathbb{P}&_{\left(\vec{B}_{i} \right)_{i=1}^{mn-k+\ell_{s}-\ell_{a}},\vec{X},\vec{Y}}\left( \tr\left( \vec{X}\vec{B}_{i}^{\top}\right)_{i} = \tr\left( \vec{Y}\vec{B}_{i}^{\top} \right)_{i} \right)\\ 
		\iftoggle{llncs}{\nonumber &= \frac{q^{\ell_{a}}-1}{q^{mn-k + \ell_{s}}-1}\left( 1 + \frac{ \frac{q^{mn-k+\ell_{s}}-1}{q^{\ell_{a}}} -1}{\sharp \mathcal{B}_{t}^{m,n,q}} \right) \\}{} 
      \nonumber  & = \frac{1}{q^{mn-k+\ell_s-\ell_{a}}} \frac{(1-\frac{1}{q^{\ell_a}})}{(1- \frac{1}{q^{mn-k+\ell_s}} )} \left( 1 + O \left( \frac{q^{mn-k +\ell_{s} - \ell_{a}}}{\sharp \mathcal{B}_{t}^{m,n,q}} \right) \right)\\                                                                                                  & = \frac{1}{q^{mn-k+\ell_s-\ell_{a}}}
                             \left( 1 + O\left(\frac{1}{q^{\ell_a}}\right) + O\left(\frac{1}{q^{mn-k+\ell_s}}\right)  + O \left( \frac{q^{mn-k +\ell_{s} - \ell_{a}}}{\sharp \mathcal{B}_{t}^{m,n,q}} \right) \right)~.\label{eq:bigohs}
    \end{align}
    Next, we have to identify which of the big-O's is the dominant
    term. Since
    $mn-k+\ell_s-\ell_a$ is the dimension of the dual code
    $\C^\perp$, this quantity is nonnegative and hence $1/q^{mn-k+\ell_s} = o(1/q^{\ell_a})
$.  Moreover, since $\Bv_1, \dots, \Bv_{mn-k}$ is
    a dual basis of a code in $\mathcal F$ \emph{i.e.} a code with minimum distance~$\geq d$,
    the assumption $t < d/2$ entails that the map below is injective
    \[
      \map{\mathcal{B}_{t}^{m,n,q}}{\Fq^{mn-k}}{\mathbf E}{{(\tr (\mathbf{E}\Bv_{i} ^\top))}_{i=1}^{mn-k}~.}
    \]
    Therefore,
    \[
      \sharp \mathcal{B}_{t}^{m,n,q} \leq q^{mn-k} \leq q^{mn-k+\ell_s}
    \]
    which shows that the dominant big-O in~(\ref{eq:bigohs}) is the rightmost one.
    Hence,
    \[
      \mathbb{P}\left( \tr\left( \vec{X}\vec{B}_{i}^{\top}\right)_{i} = \tr\left( \vec{Y}\vec{B}_{i}^{\top} \right)_{i} \right)
		= \frac{1}{q^{mn-k+\ell_s-\ell_{a}}}\left( 1 + O\left( \frac{q^{mn-k +\ell_{s} - \ell_{a}}}{\sharp \mathcal{B}_{t}^{m,n,q}} \right) \right) 
	\]
	which concludes the proof.\iftoggle{llncs}{\qed}{}  
\end{proof}

We are now ready to prove Proposition \ref{proposition:domainSampleability}.

\begin{proof}[Proof of Proposition \ref{proposition:domainSampleability}]
	We just combine Lemmas \ref{lemme:leftOver} and \ref{lemma:probaCol} and the fact that the $\left( \vec{B}_{i} \right)_{i=1}^{mn - k + \ell_{s} - \ell_{a}}$ output by Algorithm \ref{algo:trap} defines a random dual basis of a random code from  $\mathsf{AddRemove}\left( \mathcal{F},\ell_a,\ell_s \right)$ with $\mathcal{F}$ being the family of Gabidulin codes with parameters $[n,k]_{q^{m}}$ (in particular these codes have minimum distance $n-k+1$).\iftoggle{llncs}{\qed}{}  
\end{proof}

%% file: src/app_prob_rank.tex
\section{Proof of Proposition~\ref{prop:prob_rank}}\label{sec:proof_prop_prob_rank}    
    
The proof of the proposition above rests on the following lemma.
\begin{lemma}\label{lem:counting}
  Let $u,v$ be a positive integers and
  $\vec{v} \in \Fq^v \setminus \{\vec{0}\}$.  Let $s$ be an integer
  such that $u, v \geq s$.  Using Notation~(\ref{eq:sigma}), the
  number of $u\times v$ matrices of rank $s$ whose first column is
  equal to $\vec{v}$ is:
  \[
    \frac{\sigma_s^{u,v,q} - \sigma_s^{u,v-1,q}}{q^u - 1}~\cdot
  \]
\end{lemma}

\begin{proof}
  For any $\vec{v}$, denote by $\mathcal{S}_s^{u,v,q}(\vec v)$ the
  subset of $\mathcal{S}_s^{u,v,q}$ of matrices whose leftmost column
  is~$\vec v$. This yields a partition
  \[\mathcal{S}_s^{u,v,q} = \bigsqcup_{\vec v \in \Fq^u}
    \mathcal{S}_s^{u,v,q}(\vec v)~. \] Recall that $\mathbf{GL}_u(\Fq)$
  acts by multiplication on the left on $\mathcal{S}_s^{u,v,q}$ and
  acts transitively on $\Fq^u \setminus \{\vec{0}\}$. Therefore, the
  action of $\mathbf{GL}_u(\Fq)$ on $\mathcal{S}_s^{u,v,q}$ permutes
  the $\mathcal{S}_s^{u,v,q}(\vec v)$ when $\vec v$ ranges over
  $\Fq^u \setminus \{\vec{0}\}$. Thus, the $\mathcal{S}_s^{u,v,q}(\vec v)$'s
  when $\vec v$ ranges over $\Fq^u \setminus \{\vec{0}\}$ all have the same
  cardinality.

  Furthermore, $\mathcal{S}_s^{u,v,q}(\vec 0)$ is in one-to-one correspondence
  with $\mathcal{S}_s^{u,v-1,q}$.
  Therefore, for any~$\vec v \in \Fq^u \setminus \{\vec{0}\}$.
  \[
    \sharp \mathcal{S}_s^{u,v,q} = \sharp \mathcal{S}_s^{u,v,q}(\vec
    0) + (q^u-1) \sharp \mathcal{S}_s^{u,v,q} (\vec v) = \sharp \mathcal{S}_s^{u,v-1,q} + (q^u-1) \sharp \mathcal{S}_s^{u,v,q} (\vec v),
  \]
  which gives
  \[
    \sharp \mathcal{S}_s^{u,v,q} (\vec v) = \frac{\sigma_s^{u,v,q} - \sigma_s^{u,v-1,q}}{q^u - 1}~\cdot 
  \]
  \iftoggle{llncs}{\qed}{} 
  \end{proof}

\begin{remark}
  According to Lemma~\ref{lem:counting}, the cadinality of the set of
  $u\times v$ matrices of rank $s$ whose first column is a prescribed
  nonzero vector is
  \[
    \approx q^{s(u+v-s)-u}~.
  \]
\end{remark}

\begin{proof}[Proof of Proposition~\ref{prop:prob_rank}]
Denote by $\Cv_{a+1}$ the $(a+1)$--th column of $\Cv$. The total probability formula yields:
  \begin{align*}
    \Prob(|\Cv| = s)  = 
    \Prob(|\Cv| = s &~|~ \Cv_{a+1} = \vec 0)\Prob(\Cv_{a+1} = \vec 0)\\
    &+ \Prob(|\Cv| = s ~|~ \Cv_{a+1} \neq \vec 0)\Prob(\Cv_{a+1} \neq \vec 0)~.
  \end{align*}
  Since $\Cv$ is a uniformly random matrix with $b-1$ prescribed zero entries, we easily get
  \[
    \Prob(\Cv_{a+1} = \vec 0) = q^{b-1-m} \ .
  \]
  Moreover, \( \Prob(|\Cv| = s ~|~ \Cv_{a+1} = \vec 0) \) is
  nothing but the probability that a uniformly random
  $m \times (n-a-1)$ matrix has rank $s$, which yields
  \[ \Prob(|\Cv| = s ~|~ \Cv_{a+1} = \vec 0) =
    \frac{\sigma_{s}^{m,n-a-1,q}}{q^{m (n-a-1)}}~\cdot
  \]
  Next,
  \begin{align*}
    \Prob\Big(|\Cv| = s &~|~ \Cv_{a+1} \neq \vec 0\Big) \\
     &=  \sum_{\vec{w} \in \Fq^{m-b+1}\setminus \{\vec 0\}} \Prob \left( |\Cv|=s ~\Big|~ \Cv_{a+1} =
                                                                 \begin{pmatrix}
                                                                   \vec 0 \\ \vec w
                                                                 \end{pmatrix}
                                                                            \right)
                                                                            \Prob\left(
                                                                            \Cv_{a+1} =                                                                  \begin{pmatrix}
                                                                   \vec 0 \\ \vec w
                                                                 \end{pmatrix}
    \right)
    \\
    &= \sum_{\vec w} \left( \frac{\sigma_s^{m, n-a,q} - \sigma_s^{m,n-a-1,q}}{(q^m-1)}\cdot \frac{1}{q^{m(n-a-1)}} \right)\Prob\left(
                                                                            \Cv_{a+1} =                                                                  \begin{pmatrix}
                                                                   \vec 0 \\ \vec w
                                                                 \end{pmatrix}
    \right) \\
                                                               & = \frac{\sigma_s^{m, n-a,q} - \sigma_s^{m,n-a-1,q}}{(q^m-1)}\cdot \frac{1}{q^{m(n-a-1)}}~\cdot
  \end{align*}
  The second equality is due to  Lemma~\ref{lem:counting} and the third one
  due to the uniformity of $\Cv$ which entails that
  \[
  \Prob\left(
                                                                            \Cv_{a+1} =                                                                  \begin{pmatrix}
                                                                   \vec 0 \\ \vec w
                                                                 \end{pmatrix}
                                                               \right)
                                                             \]
  does not depend on $\vec w \in \Fq^{m-b+1}$.
  Putting all together yields
  \[
    \Prob(|\Cv| = s) =     \frac{\sigma_{s}^{m,n-a-1,q}}{q^{m(n-a-1)}} q^{b-1-m}  +
    \frac{\sigma_s^{m, n-a,q} - \sigma_s^{m,n-a-1,q}}{(q^m-1)q^{m(n-a-1)}} (1-q^{b-1-m}) \ .
  \]
  Asymptotically the left and right--hand summands are approximately equal to
  \[
    q^{s(m+n-a-s) + m(a-n) + b-1-s} \quad \text{and} \quad
    q^{s(m+n-a-s) + m(a-n)}~.
  \]
  Thus, the dominant term is
  \(
    \approx q^{s(m+n-a-s)  + m(a-n) + \max (0, b-1-s)}.
    \) \iftoggle{llncs}{\qed}{}
\end{proof}